\documentclass[11pt]{article}

\usepackage[letterpaper, margin=.95in]{geometry}

\usepackage{amsmath, amsthm, amssymb}
\usepackage{cite}
\usepackage{url}
\usepackage{cleveref}
\usepackage{appendix}

\usepackage{times}

\usepackage{color}
\usepackage{algorithm}
\usepackage[noend]{algpseudocode}
\usepackage{changepage}

\usepackage{graphicx}
\usepackage{epstopdf}
\usepackage{multirow}

\usepackage{caption}

\usepackage{enumitem}
\setitemize{noitemsep, topsep=0pt,parsep=0pt,partopsep=0pt}
\algnewcommand\algorithmicswitch{\textbf{switch}}
\algnewcommand\algorithmiccase{\textbf{case}}

\algdef{SE}[SWITCH]{Switch}{EndSwitch}[1]{\algorithmicswitch\ #1\ \algorithmicdo}{\algorithmicend\ \algorithmicswitch}%
\algdef{SE}[CASE]{Case}{EndCase}[1]{\algorithmiccase\ #1}{\algorithmicend\ \algorithmiccase}%
\algtext*{EndSwitch}%
\algtext*{EndCase}%

\newcommand{\eps}{\varepsilon}

\newcommand{\floor}[1]{\lfloor #1 \rfloor}

\newcommand{\NumOfLayers}{L}
\newcommand{\set}[1]{\left\{#1\right\}}

\DeclareMathOperator{\E}{\mathbb{E}}

\newcommand{\layer}{\mathit{layer}}
\newcommand{\CDS}{\mathit{CDS}}

\newtheorem{theorem}{Theorem}[section]
\newtheorem{lemma}[theorem]{Lemma}

\newtheorem{proposition}[theorem]{Proposition}
\newtheorem{observation}[theorem]{Observation}
\newtheorem{definition}[theorem]{Definition}

\newtheorem*{theoremRinner}{Theorem~\ref{\repthmref}}

\newenvironment{theoremR}[1]
  {\def\repthmref{#1}\theoremRinner (restated)}{\endtheoremRinner}
	\newenvironment{theoremRF}[1]
  {\def\repthmref{#1}\theoremRinner (restated and rephrased)}{\endtheoremRinner}

\newtheorem*{observationRinner}{Observation~\ref{\repthmref}}
\newenvironment{observationR}[1]
  {\def\repthmref{#1}\observationRinner (restated)}{\endobservationRinner}

\usepackage{theomac} %To repeat theorems
\newtheoremWithMacro{rtheorem}[theorem]{Theorem}

\newcommand{\FullOrShort}{short}

\ifthenelse{\equal{\FullOrShort}{full}}{
	
	  \newcommand{\fullOnly}[1]{#1}
	  \newcommand{\shortOnly}[1]{}
		\newcommand{\IncludePictures}[1]{#1}

  }{

	  \newcommand{\fullOnly}[1]{}
	  \newcommand{\shortOnly}[1]{#1}
		\newcommand{\IncludePictures}[1]{}

  }

\begin{document}

\date{}

\title{A New Perspective on Vertex Connectivity}

\author{
  Keren Censor-Hillel\\
  \small Technion\\
  \small Haifa, Israel\\
  \texttt{\small ckeren@cs.technion.ac.il}
  \and 
  Mohsen Ghaffari\\
  \small MIT\\
  \small Cambridge, MA, USA\\
  \texttt{\small ghaffari@mit.edu}
  \and
  Fabian Kuhn\\
  \small University of Freiburg\\
  \small Freiburg, Germany\\
  \texttt{\small kuhn@cs.uni-freiburg.de}
}

\maketitle

\begin{abstract}
Edge connectivity and vertex connectivity are two fundamental concepts in graph theory. Although by now there is a good understanding of the structure of graphs based on their edge connectivity, our knowledge in the case of vertex connectivity is much more limited. An essential tool in capturing edge connectivity are edge-disjoint spanning trees. The famous results of Tutte and Nash-Williams show that a graph with edge connectivity $\lambda$ contains $\floor{\lambda/2}$ edge-disjoint spanning trees.

We present connected dominating set (CDS) partition and packing as tools that are analogous to edge-disjoint spanning trees and that help us to better grasp the structure of graphs based on their vertex connectivity. The objective of the CDS partition problem is to partition the nodes of a graph into as many connected dominating sets as possible. The CDS packing problem is the corresponding fractional relaxation, where CDSs are allowed to overlap as long as this is compensated by assigning appropriate weights.  CDS partition and CDS packing can be viewed as the counterparts of the well-studied edge-disjoint spanning trees, focusing on vertex disjointedness rather than edge disjointness. We constructively show that every $k$-vertex-connected graph with $n$ nodes has a CDS packing of size $\Omega(k/\log n)$ and a CDS partition of size $\Omega(k/\log^5 n )$. We moreover prove that the $\Omega(k/\log n)$ CDS packing bound is existentially optimal.

CDS packing allows us to analyze vertex connectivity in the of context random vertex sampling. We show that if vertices of a $k$-vertex-connected graph are independently sampled with probability $p$, then the graph induced by the sampled vertices has vertex connectivity $\tilde{\Omega}(kp^2)$. This bound is optimal up to polylogarithmic factors. 

As an additional application, we also show that CDS packings are tightly related to the (throughput of) store-and-forward algorithms in the networking model where in each time unit, each node can send one bounded-size message to all its neighbors. As a consequence, our $\Omega(k/\log n)$ CDS packing construction yields a store-and-forward broadcast algorithm with optimal throughput.
\end{abstract}

\thispagestyle{empty}
\clearpage
\setcounter{page}{1}

%\section*{Title Ideas:}
%
%\begin{itemize}
%\item Robustness and Broadcast Capacity of Graphs with Large Vertex
%Connectivity
%\item Robustness and Broadcast Capacity of k-Vertex-Connected Graphs
%\item Robustness, Broadcast Capacity, and the Vertex Connectivity of
  %Graphs
%\item A New Perspective on Vertex Connectivity
%\end{itemize}

\section{Introduction and Related Work}

Vertex and edge connectivity are two core graph-theoretic concepts as
they are basic measures for the robustness and throughput capacity of
a graph. While by now a lot is known about edge connectivity and its
connections to related graph-theoretic properties and problems, our
knowledge about vertex connectivity is much scarcer and many related problems remain open.

As an example, given a graph $G$, assume that each edge or node is
independently sampled with probability $p$. How large should $p$ be
such that the subgraph given by the sampled edges or the one induced
by the sampled nodes is connected (this is sometimes also known as the
\emph{network reliability}) or such that the sampled subgraph
satisfies some other properties. Intuitively, the larger the
connectivity of $G$ is, the smaller we should be able to choose $p$
such that the sampled subgraph remains connected. For edge
connectivity and sampling edges, Lomonosov and Poleskii \cite{LP71}
verified this intuition already four decades ago: if $p
=\Omega(\frac{\log n}{\lambda})$, where $\lambda$ is the
edge-connectivity of $G$, then the edge-sampled graph is connected
with high probability and this threshold is optimal. In the special
case of sampling edges of a complete graph, this corresponds to the
$\frac{\ln n}{n}$ probability threshold for connectivity in the
Erd\H{o}s-R\'{e}nyi random graph model. Karger \cite{KargerSTOC94}
showed that assuming $p =\Omega(\frac{\log n}{\lambda})$, the
edge-connectivity of the edge-sampled graph will be around $\lambda
p$, w.h.p., and in fact, for such $p$, the size of each edge cut
remains around its expectation. In the following years, these results,
and extensions thereof, have emerged as powerful tools, having
implications for numerous important problems, see
e.g. \cite{KargerSTOC94, KargerSODA94, KargerSTOC95, BenczurKarger96,
  KargerSTOC98}
%graph sparsifiers\cite{BenczurKarger96}, finding or approximating minimum edge
%cuts\cite{KargerSTOC94, KargerSODA94}, minimum s-t edge cuts\cite{KargerSTOC94, BenczurKarger96}, or max-flow in undirected
%graphs\cite{KargerSTOC94, KargerSTOC98}, network design problems\cite{KargerSTOC94}, or approximating the
%network reliability\cite{KargerSTOC95}.

In contrast, prior to our work, for vertex sampling (and vertex connectivity), even the most
basic of these questions remained open. In the present paper, we
prove results of the same flavor as the ones discussed above, but in
the context of vertex sampling rather than edge sampling. In
particular, we show that if each node of a $k$-vertex connected graph
$G$ is independently sampled with probability $p$, then w.h.p., the
graph induced by the sampled nodes has vertex connectivity at
least $\tilde{\Omega}(k p^2)$. We also show that this is existentially
tight up to $\log$-factors.\footnote{\label{footnote1}For exact statements and a more
  extensive discussion of our results, we refer to Section
  \ref{subsec:results}.} 
% In addition, we also show that if each edge of
% a $k$-vertex connected graph $G$ is sampled with probability $p$,
% w.h.p., the graph induced by the remaining set of edges has vertex
% connectivity $\tilde{\Omega}(kp)$.

The main hurdle on the way to proving these results is that there
can be an exponential number of \emph{``small''} vertex cuts. When
arguing that the subgraph induced by a randomly chosen subset of nodes
is connected, one essentially needs to show that for each vertex cut
of the graph, at least one node is selected. % However, a graph has
% exponentially many vertex cuts, which may be in arbitrary complex
% configurations with respect to each other.
However, it has been shown that even the number of the minimum vertex
cuts of a $k$-vertex connected graph can be as large as $\Theta(2^{k}
(n/k)^2)$~\cite{KanevskySODA93}. Note that this is in stark contrast to the
case of edge cuts, where the number of minimum edge cuts is known to
be bounded by $O(n^2)$ and the number of edge cuts of size $\alpha
\cdot \lambda$ in a graph with edge-connectivity $\lambda$ is at most
$O(n^{2\alpha})$~\cite{KargerSODA93, KargerSTOC96}. This $O(n^{2\alpha})$ bound is the main tool in studying edge cuts after
random sampling. \footnote{We remark that the (tightest) proofs of this bound use the edge-disjoint spanning tree results of Tutte and Nash-Willams \cite{KargerSTOC96}.} 

% ; using this bound, a simple Chernoff bound on each cut
% and then a union bound over all the cuts is sufficient to show that
% all edge cuts remain around their expectation when we sample
% edges~\cite{KargerSTOC94}. Since there can be an exponential number of
% small vertex cuts, and because the structure of vertex cuts appears
% more complex in general, standard techniques such as, e.g.,
% considering each cut individually and then applying a union bound over
% all the cuts do not work when sampling nodes instead of edges.

% In this work, we present novel techniques to break this barrier.
Our main technical contribution is a method to decompose a graph into
(almost) disjoint connected dominating sets (CDSs). We define a
\emph{CDS partition} of size $K$ as a partition of the nodes of a
graph into $K$ connected dominating sets. 
% The maximum possible size of
% a CDS partition of a graph $G$ is called the \emph{connected domatic
%   number} of $G$. 
We also define \emph{CDS packings} as a natural
generalization of CDS partitions. A CDS packing is a collection of
CDSs with positive weights such that for each node $v$, the sum of the
weights of all CDSs containing $v$ is at most $1$. The size of a CDS
packing is the total weight of all CDSs in the collection.
%  and
% analogously to before, the maximum possible size of a CDS packing of
% $G$ is called the \emph{fractional connected domatic number} of $G$.
We show that every $k$-vertex connected graph $G$ has a CDS packing of
size at least $\Omega(k/\log n)$. We also show that this is optimal in
the sense that for all $n$ and $k$, there are $k$-vertex connected
$n$-node graphs for which the largest CDS packing has size at most
$O(k/\log n)$. A generalized version of the CDS packing upper bound is
the basis for the random sampling result mentioned above and also for
a CDS partition upper bound showing that every $k$-vertex connected
graph has a CDS partition of size at least $\Omega(k/\log^5
n)$.
%\footnotemark[\ref{footnote1}]

CDS partitions and CDS packings can be seen as the ``vertex world''
analogues of the well-studied edge-disjoint spanning trees. Note that
CDS partitions and CDS packings can equivalently be seen as
collections of vertex disjoint dominating trees and dominating tree
packings, respectively, by removing cycles after solving the
problem. Edge-disjoint spanning trees have been a classical problem in
graph theory (also studied under the title of decomposing a graph into
connected factors), and they have numerous applications for different
problems concerning edge-connectivity (see e.g., \cite{KargerSTOC94}).
% For example, they can be used as tools for
% tightly counting the number of minimum edge cuts of a graph \cite{} or
% for finding minimum edge cuts\cite{}.
In a graph with edge connectivity $\lambda$, the size of a maximum
spanning tree packing is clearly at most $\lambda$. Using the famous
results of Tutte and Nash-Williams~\cite{Tutte,Nash-Williams},
Kundu~\cite{Kundu} showed that the size of a maximum set of
edge-disjoint spanning trees is at least
$\floor{\frac{\lambda}{2}}$. Thus there is no asymptotic gap between
the size of the best edge-disjoing spanning tree collection (or
packing) and the edge connectivity. Our results show the corresponding
relation between vertex connectivity and dominating tree packing and
partition: the gap for dominating tree packing is $\Theta(\log n)$,
whereas the gap for dominating tree partition is lower bounded by
$\Omega(\log n)$ and upper bounded by $O(\log^5 n)$.

Apart from being interesting and natural structures of their own, CDS
partition and CDS packing also have interesting applications in
communication networks. Assume that in every time unit, each node of a
network can locally broadcast a message of at most $B$ bits to all its
neighbors. We show that the achievable total throughput when globally
broadcasting messages using store-and-forward routing
algorithms\footnote{Store-and-forward algorithms constitute the
  classical paradigm of message routing, where network nodes only
  forward the messages they receive from their neighbors, possibly
  with some changes in header information, without the use of network
  coding or any other form of combining different messages into new
  ones.}  can exactly be characterized by the size of the largest CDS
packing of the network graph. As a consequence, we get that in
$k$-vertex connected networks, such algorithms can achieve an optimal
throughput of $\Theta(k/\log n)$ messages per round. Techniques
of~\cite{HaeuplerSTOC11} show that network coding can achieve a
throughput of $\Theta(k)$. Thus, our results imply that the network
coding advantage for simultaneous broadcasts is $\Theta(\log n)$. Note
that determining the network coding advantage for different
communication models and problems is one of the important questions
when studying network coding (see, e.g.,~\cite{Li-Li-Lau,
  agarwal2004advantage, Langberg, Goel}).

\subsection{Results}\label{subsec:results}

To cope with the problem of facing exponentially many small vertex
cuts, we use a \emph{layering} idea. For all our results, we need to
find or prove the existence of a collection of small connected dominating
sets of some graph $H$ with large vertex connectivity. Whereas in all
cases, domination will be straightforward, obtaining connectivity is
much more challenging. We assume that the nodes of $H$ are partitioned
into $O(\log n)$ layers. We go through the layers one-by-one and
establish connectivity by growing components as we proceed. Growing
components turns out easier than proving connectivity directly and we
can show that using this step-wise approach, it suffices to consider
only polynomially many vertex cuts of $H$.

% To cope with the problem of facing exponentially many small vertex
% cuts, we use a \emph{layering} idea which is easiest described in the
% context of sampling. Assume that the nodes of a graph are partitioned
% into $O(\log n)$ layers and assume that every node of the graph is
% randomly sampled with probability $p$.  We work through the layers one
% by one. Let $S_\ell$ be the sampled nodes of layer $\ell$. For some
% layer $\ell$, consider the connected components of the graph induced
% by $S_1\cup\dots\cup S_\ell$. We show that if the node sampling
% probability $p$ is large enough and the graph induced by all the layer
% $\ell+1$ nodes has sufficiently large vertex connectivity, then when
% adding the sampled layer $\ell+1$ nodes $S_{\ell+1}$, in expectation,
% a constant fraction of the components of the subgraph induced by
% $S_1\cup\dots\cup S_\ell$ gets connected to at least one other
% component. Consequently, after going through all $O(\log n)$ layers,
% the graph induced by the sampled nodes is connected. Carried out
% carefully, the described approach allows to prove our most basic
% sampling result.

\subsubsection*{CDS Packing and CDS Partition}

As described, our core technical contribution is an efficient
algorithm to construct large CDS packings. As any connected dominating
set of a graph $G$ must contain at least one node from every vertex
cut of $G$, it is not hard to see that any CDS partition of a
$k$-vertex connected graph has size at most $k$. Using a similar
argument, it also follows that the largest CDS packing of a $k$-vertex
connected graph has size at most $k$. Our main result shows that this
basic upper bound can almost be achieved.

\begin{theorem}\label{thm:packing} Every $k$-vertex-connected $n$-node graph
  has a CDS packing of size $\Omega\big(\frac{k}{\log n}+1\big)$.
\end{theorem}

Specifically, we show how to construct a collection of $k$ CDSs, each consisting of $O\big(\frac{n\log n}{k}\big)$ nodes, such that each node is in at most
$O(\log n)$ of the CDSs. Based on the CDS packing construction, we
also obtain an efficient way to get a large CDS partition leading to
the following result.

\begin{theorem}\label{thm:partition} Every $k$-vertex-connected $n$-node graph
  has a CDS partition of size $\Omega\big(\frac{k}{\log^5 n}+1\big)$.
\end{theorem}

In addition, we complement these results by showing that our $\Omega\big(\frac{k}{\log n}\big)$ CDS
packing bound is asymptotically optimal in general. 

\begin{theorem} \label{thm:BadExample} For any sufficiently large $n$,
  and any $k$, there exist $n$-node graphs with vertex connectivity $k
  \geq 1$ where the maximum CDS packing size, and thus also the maximum CDS partition size, are $O\big(\frac{k}{\log
    n}+1\big)$.
\end{theorem}

\subsubsection*{Vertex Connectivity and Random Sampling:}

As a specific application of (a generalized version of) our CDS
packing construction, we obtain a tool to analyze the graph that is
obtained when randomly sampling a subset of the nodes of a graph with
large vertex connectivity. Mainly, we prove a lower bound on the
vertext connectivity of the graph induced by a set of randomly sampled
nodes. Note that in the following, for a graph $G=(V,E)$ and set of
nodes $S\subseteq V$, $G[S]$ denotes the subgraph of $G$ induced by
$S$.

\begin{theorem}\label{thm:general}
  Consider a $k$-vertex-connected, $n$-node graph $G=(V,E)$ and let
  $S$ be a subset of $V$ where each node $v\in V$ is included in $S$
  (i.e., \emph{sampled}) independently with probability $p$. W.h.p.,
  the graph $G[S]$ has vertex-connectivity $\Omega\big(\frac{k p^2}{\log^3
    n}\big)$.\footnote{Note that the theorem requires $k=\Omega(\log^3 n)$
    to be meaningful. Such a polylogarithmic lower bound on the vertex
    connectivity is necessary for all our statements to become
    non-trivial.  We would like to point out that this is in all cases
    necessary as, e.g., shown by Observation \ref{obs:neg-side}.}
\end{theorem}

A simple intuitive argument shows that up to logarithmic factors, the
statement of \Cref{thm:general} is the best possible. Consider a graph
consisting of two cliques of size $k$, connected via a matching of $k$
edges. When randomly choosing each node with probability $p$, the
vertex connectivity of the graph induced by the chosen nodes is given
by the number of surviving matching edges. As each of these edges
survives with probability $p^2$ and there are $k$ such edges, the
expected vertex connectivity after sampling, even in this simple graph, is $kp^2$. Theorem
\ref{thm:general} also gives a lower bound on $p$ such that the
induced graph of the sampled nodes is connected, w.h.p. A slightly
tighter bound for this specifc case is given by the following
theorem. 

\begin{theorem}\label{thm:threshold}
  Consider a $k$-vertex-connected, $n$-node graph $G=(V,E)$ and let
  $S$ be a subset of $V$ where each node $v\in V$ is included in $S$
  (i.e., \emph{sampled}) independently with probability $p$. If $p >
  \alpha\frac{\log n}{\sqrt{k}}$ for a sufficiently large constant
  $\alpha$, then $G[S]$ is connected, w.h.p.
\end{theorem}

We prove \Cref{thm:threshold} in Section \ref{sec:threshold}. Apart
from being interesting by itself, we use this theorem to introduce the
basic proof structure that we also use for our other, more involved
results. The following observation shows that for most values of $k$,
\Cref{thm:threshold} is tight up to a factor of $O(\sqrt{\log n})$ (a
formal argument appears in \Cref{app7}).

\begin{observation}\label{obs:neg-side} For every $k$ and $n$, there
  exists an $n$-node graph $G$ with vertex connectivity $k$ such that
  if we independently sample vertices with probability $p\leq
  \frac{\sqrt{\log(n/2k)}}{\sqrt{2k}}$, then the subgraph induced by the
  sampled nodes is disconnected with probability at least $1/2$.
\end{observation}

% Our proof of \Cref{thm:general} has a relatively simple structure and
% only uses basic graph-theoretic facts (e.g. Menger's theorem) and
% standard probabilistic arguments. Moreover, it turns out to be quite
% flexible; for instance, with small modifications, we also get bounds
% for the remaining vertex connectivity after random edge sampling:

% \begin{theorem}\label{thm:VC_ES} Consider an arbitrary graph $G=(V,E)$ with
%   vertex connectivity $k = \Omega(\log^3 n)$. Let $E'$ be a subset of
%   $E$ where each edge $e\in E$ is included in $E'$ (i.e.,
%   \emph{sampled}) independently with probability $p$. With high
%   probability, the graph $H=(V, E')$ has vertex-connectivity
%   $\Omega(\frac{k p}{\log^3 n})$.
% \end{theorem}

% The proof of \Cref{thm:VC_ES} is similar to the proof of
% \Cref{thm:general} and it is thus appears in \Cref{sec:VC_ES}.

\subsection{Additional Related Work}

The domatic number of a graph is the size of the largest partition of
a graph into dominating sets. In \cite{feige02} it is shown that for
graphs with minimum degree $\delta$, nodes can be partition into $(1-o(1))(\delta+1)/\ln n$ dominating sets
efficiently. This implies a $(1+o(1))\ln n$-approximation, which is
shown to be best possible unless
$\mathrm{NP}\subseteq\mathrm{DTIME}(n^{O(\log\log n)})$. Further,
Hedetniemi and Laskar \cite{domination} present an extensive
collection of results revolving around dominating sets. The CDS
partition problem was first introduced in~\cite{HedetniemiLaskarCDN}
where the size of a maximum CDS partition of a graph $G$ is called the connected domatic number of $G$. Zelinka~\cite{Zelinka}
shows a number of results about the connected domatic number; in
particular, that it is upper bounded by the vertex connectivity. \cite{Hartnell-Rall} shows that the connected domatic
number of planar graphs is at most $4$ and also describes
some relations between the number of edges of a graph and its
connected domatic number. Finally, \cite{RotatingCDS} argues
that a large CDS partition can be useful for balancing energy usage in
wireless sensor networks.
%  and present a distributed heuristic which achieves a CDS
% partition of size $\lfloor \frac{\delta+1}{\beta\cdot(c+1)}\rfloor -
% \eps \delta |V|$ where $\delta$ is the minimum node degree, and
% $\beta$, $c$, and $\eps \ll 1$ are constants.

\subsection{Roadmap}
The rest of the paper is organized as follows. In \Cref{sec:prelim},
we formally define CDS partition and packing and describe their
connections to vertex connectivity, as well as their networking
applications. In \Cref{sec:threshold}, we present the result about the
threshold on the vertex sampling probability for getting a connected
induced subgraph. The analysis there demonstrate the main structure of
our analysis for the later results in Sections \ref{sec:general} and
\ref{sec:partition}. In \Cref{sec:general}, we present the CDS packing
construction which is also used to derive our general bound on the
vertex connectivity after vertex sampling. In \Cref{sec:partition}, we
modify this construction to get the claimed CDS partition
results. Finally, \Cref{sec:discussions} presents some concluding
remarks.

%%% Local Variables: 
%%% mode: latex
%%% TeX-master: "main_FOCS"
%%% End: 

\section{Vertex Connectivity and Connected Dominating Sets}
%\section{Introducing CDS Partition and CDS Packing}
\label{sec:prelim}

In this section, we formally introduce CDS partitions and CDS
packings. The structures can be seen as a well-organized way to
capture the vertex connectivity of a graph $G$. We therefore believe
that they also provide a new perspective on vertex connectivity and
related problems. The section is organized as follows. We start out by
formally defining all necessary concepts. We then describe how CDS
partition and CDS packing relate to vertex connectivity. Finally, we
conclude the section by describing an application of the two
structures in a networking context.

% ``skeletons'' that capture
% the vertex connectivity in an well-organized format, thus also
% providing a new viewpoint for vertex connectivity related problems. We
% first start with introducing some notions and defining CDS partition
% and CDS packing. Then, we explain how these structures relate to
% vertex connectivity. Finally, we explain how these structures can be
% useful in the networking world.

\subsection{Definitions}
\paragraph{Notations} We usually work with a simple undirected graph
$G=(V,E)$ as our main graph. Troughout, we use $n$ for the number of
nodes of $G$. Also, for a subset $S \subseteq V$ of the nodes, we use
$G[S]$ to denote the subgraph of $G$ induced by $S$.
 
\begin{definition}[Dominating Set and Connected Dominating Set] Given
  a graph $G = (V, E)$, a set $S\subseteq V$ is called a dominating
  set iff each node $u \in V\setminus S$ has a neighbor in $S$. The
  set $S$ is called a connected dominating set (CDS) iff $S$ is a
  dominating set and $G[S]$ is connected.
\end{definition}
If $S$ is a dominating set of $G=(V,E)$, we also say that $S$
dominates $V$.

\begin{definition}[CDS Partition]  
  A \emph{CDS partition} of a graph $G=(V,E)$ is a partition $V_1
  \cup\dots\cup V_{t}=V$ of the nodes $V$ such that each set $V_i$ is
  a CDS. The \emph{size} of a CDS partition is the number of CDSs of
  the partition. The maximum size of a CDS partition of $G$ is
  denoted by $K_{\CDS}(G)$.
\end{definition}

\Cref{fig:CDS partition} presents a graph with a CDS partition of size
$2$, where nodes of each color form a CDS.

\begin{figure}[t]
  \centering
  \includegraphics[width=50 mm]{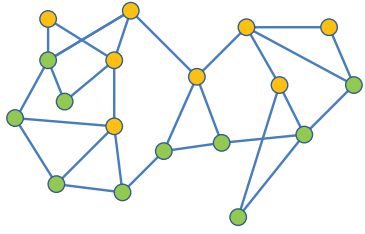}
  \caption{A CDS partition of size $2$}
  \label{fig:CDS partition}
	\vspace{-0.15cm}
\end{figure}

\begin{definition}[CDS Packing]
  Let $\CDS(G)$ be the set of all CDSs of a graph $G$. A \emph{CDS packing} of $G$
  assigns a non-negative weight $x_{\tau}$ to each $\tau \in \CDS(G)$
  such that for each node $v \in V$, {$\sum_{\tau \ni v }x_\tau\leq
    1$.}
  The \emph{size} of this CDS packing is $\sum_{\tau \in \CDS(G)}
  x_{\tau}$. The maximum size of a CDS packing of $G$ is denoted by
  $K'_{\CDS}(G)$.
\end{definition}
Note that a CDS partition is a special case of a CDS packing where
each $x_{\tau} \in \{0, 1\}$. In other words, CDS packing is the LP
relaxation of CDS partition when formulating CDS partition as an
integer programming problem in the natural way. Consequently, we have
$K_{\CDS}(G) \leq K'_{\CDS}(G)$ for every graph $G$.

We remark that the maximum CDS partition size of graph $G$ is
sometimes also called the \emph{connected domatic number} of
$G$\cite{HedetniemiLaskarCDN, Zelinka, Hartnell-Rall}. Analogously,
the maximum CDS packing size can be referred to as the
\emph{fractional connected domatic number} of $G$.

\subsection{CDS Packing and Vertex
  Connectivity}\label{subsec:packVSConn}

Menger's theorem tells us that a graph is $k$-vertex connected if and
only if every pair of nodes $u$ and $v$ can be connected through $k$
internally vertex-disjoint paths. A CDS partition or packing produces
analogous systems of paths, but with much stronger requirements: Given
a CDS partition of size $t$, each pair of nodes $u$ and $v$ is
connected by $t$ internally vertex-disjoint paths, one through each
CDS. Consequently, we can choose the vertex-disjoint paths for
different pairs in a consistent way in the following sense. The paths
for all the pairs can be colored by $t$ colors such that, paths of
different colors are internally vertex-disjoint and each node pair can
be connected by one path from each color. A CDS packing provides a
relaxed version of such a path system, where for each pair, paths are
allowed to have vertex-overlaps but the weighted overlap in each node
is bounded by one. These observations closely resemble edge-disjoint
spanning trees (see e.g. \cite{Tutte, Nash-Williams, Kundu}). A set of
$t$ edge-disjoint spanning trees gives rise to a similar colored
system of edge-disjoint (rather than internally vertex-disjoint)
paths.

In light of the above discussion, it is natural to ask how the maximum
CDS packing or partition sizes compare with the vertex connectivity of
a graph. One direction of this relation is straightforward. Since for
each vertex cut, each CDS has to contain at least one node of this
cut, we have $K_{\CDS}(G) \leq k$ \cite{Zelinka, Hartnell-Rall}. Based
on the same basic argument, the same upper bound also applies to CDS
packings (proof in \Cref{sec:BadExample}):

\begin{proposition}\label{prop:upp} For each graph with
  vertex-connectivity $k$, we have $K_{\CDS}(G) \leq K'_{\CDS}(G) \leq
  k$.
\end{proposition}

The other direction of the relation is more interesting. In the ``edge
world", the famous Tutte and Nash-Williams result \cite{Tutte,
  Nash-Williams} shows that each graph with edge connectivity
$\lambda$ contains at least $\lambda/2$ edge-disjoint spanning trees,
creating a colored system of edge-disjoint paths as explained above
without asymptotic loss compared to Menger's Theorem. Our main result
in this respect is \Cref{thm:packing}, which proves a corresponding
statement for the ``vertex world" by showing that any $k$-vertex
connected graph has a CDS packing of size $\Omega(k/\log n)$. Our
proof of \Cref{thm:packing}, presented in \Cref{sec:general}, directly
leads to an efficient construction of a CDS packing of size
$\Omega(k/\log n)$. As argued, CDS packings thus offer a new
structured way of looking at vertex connectivity. Unlike for edge
connectivity, the additional structure comes at a price as we loose an
$O(\log n)$-factor, which we show to be unavoidable (in
\Cref{thm:BadExample}).

As a testament to the strength of the structure, in the proof of
\Cref{thm:general}, we explain that the same general construction
produces a CDS packing of size $\tilde{\Omega}(kp^2)$ when nodes are
sampled with probability $p$. Up to logarithmic factors, this solves
the open question about the vertex connectivity after sampling. The
CDS packing acts as a witness for the vertex connectivity of the
sampled subgraph, showing that it is $\tilde{\Omega}(kp^2)$.

Also, in \Cref{sec:partition} we show that a similar construction produces a CDS partition of size $\Omega(\frac{k}{\log^5 n})$. We find it interesting that this CDS
partition construction itself uses \Cref{thm:general}, i.e., the sampling result.

\subsection{CDS Packing vs. Throughput}\label{subsec:packVSthru}
To conclude Section \ref{sec:prelim}, we explain that CDS packing and
CDS partition are relevant structures, also from a networking point of
view, as they are closely related to the throughput of
store-and-forward algorithms. A store-and-forward algorithm
corresponds to the classical paradigm of routing (in contrast to
network coding) where each node only stores and forwards packets and does not combine messages, or parts of them. Consider a
synchronous network model where in each round (time unit), each node
can send one message of size at most $B$ bits to all of its neighbors. Such
a communication model is motivated, e.g., when considering wireless
networks and working above the MAC layer, i.e., a local broadcast
layer (see e.g. \cite{AbsMAC}). In the described communication model,
CDS packing perfectly captures the throughput of store-and-forward
algorithms for concurrent global broadcasts:

\begin{theorem}\label{prop:thruConnection} A CDS packing with size $t$
  provides a store-and-forward backbone with broadcast throughput
  $\Omega(t)$ messages per round. Inversely, a store-and-forward
  broadcast algorithm with throughput $t$ messages per round induces a
  CDS packing of size at least $t$.
\end{theorem}

Here we provide a brief intuition. The formal proof of
\Cref{prop:thruConnection} is deferred to \Cref{sec:App5}. For the
first part, consider a CDS partition of size $t$. Then, $t$ different
messages can be routed along different CDSs simultaneously and thus we
can concurrently make progress for $t$ messages (throughput $t$). To
achieve this using a CDS packing, each node $v$ is time-shared between
the CDSs which go through $v$ such that the duration that $v$ works
for each CDS $\tau$ is proportional to the weight $x_{\tau}$. For the
second part, consider a store-and-forward broadcast algorithm with
throughput $t$ and run it using a sufficiently large number of
messages. For each message $m$, the set of nodes that forward $m$
forms a CDS (otherwise, $m$ would not reach all the nodes). Choosing
the weight of each given CDS proportional to the number of messages
that use this CDS induces a CDS packing of size $t$.

Given this connection, we get that our CDS packing result
(\Cref{thm:packing}) also gives a store-and-forward broadcast
algorithm with optimal throughput. We remark that in the general
formulation of CDS packings, each node might participate in
arbitrarily many (in fact up to exponentially many) CDSs. This would
make CDS packing inefficient from a practical point of view if the
number of messages is small compared to the number of CDSs
used. Fortunately, in our construction (cf., \Cref{thm:packing}), each
node only participates in $O(\log n)$ CDSs, which makes the CDS
packing efficient even for a small number of messages.

%%% Local Variables: 
%%% mode: latex
%%% TeX-master: "main_FOCS"
%%% End: 

\section{Vertex Connectivity Under Vertex Sampling: Connectivity Threshold}
\label{sec:threshold}
In this section, we study the threshold for the vertex sampling
probability such that the sampled graph \emph{remains connected}. This
allows us to demonstrate the main outline of our analysis technique in
a simple and clean way. In the next section, we extend the technique
to prove our general bounds on CDS packing and on the vertex
connectivity after sampling.

\begin{theoremR}{thm:threshold}
  Consider an arbitrary graph $G=(V,E)$ with vertex connectivity
  $k$. Let $S$ be a subset of $V$ where each node $v\in V$ is included
  in $S$ (i.e., \emph{sampled}) independently with probability $p$. If
  $p > \alpha\frac{\log n}{\sqrt{k}}$ for a sufficiently large
  constant $\alpha$, then, w.h.p., $S$ is a CDS of $G$ and thus,
  $G[S]$ is connected.
\end{theoremR}

Before proceeding with the actual proof of \Cref{thm:threshold}, we
need to introduce a few basic concepts. As a key feature of our
analysis, instead of trying to argue about the structure of the set of
all the sampled nodes at once, we turn the sampling into a more
evolving process, using a simple \emph{layering} idea. To show that a
set of randomly sampled nodes of a graph forms a CDS, we partition the
nodes of the graph into $\NumOfLayers$ layers and study how the
sampled structure evolves when going through the layers one-by-one. In
order for our arguments to work, we require the subgraph induced by
each layer to have large vertex connectivity. This is hard to achieve
when partitioning the nodes of the original graph $G$ into layers. Instead
of arguing directly about $G$, we therefore analyze the sampling
process for a graph $\mathcal{G}=(\mathcal{V},\mathcal{E})$ that we
call the \emph{virtual graph}\footnote{We remark here that this transformation to virtual graph $\mathcal{G}$ is not necessary in \Cref{sec:threshold} but it is needed in \Cref{sec:general}. In order to make the analysis of the two sections consistent and to get the readers used to the related notations before going into the complications of \Cref{sec:general}, we choose to use $\mathcal{G}$ in this section as well.} and which is defined as follows. For
each real node $v\in V$ (and a sufficiently large constant $\lambda$),
create $\NumOfLayers=\lambda \log n$ copies of $v$, one for each layer
$\ell$ in $[1, \NumOfLayers]$. Connect two virtual nodes if and only
if they are copies of the same real node or copies of two adjacent
real nodes. Note that for each layer $\ell$, the virtual nodes of
layer $\ell$ induce a copy of $G$. For each set of virtual nodes
$\mathcal{W} \subseteq \mathcal{V}$, we define the projection
$\Psi(\mathcal{W})$ of $\mathcal{W}$ onto $G$ as the set $W \subseteq
V$ of real nodes $w$, for which at least one virtual copy of $w$ is in
$\mathcal{W}$. Note that two nodes in $\mathcal{G}$ are connected if
and only if their projections are connected in $G$ or if they project to the same node, implying that
$\mathcal{G}[\mathcal{W}]$ is connected (or dominating) if and only if
$G[\Psi(\mathcal{W})]$ is connected (or dominating).

To translate the sampling to $\mathcal{G}$, we use a simple coupling
argument. Consider the following process: sample each virtual node
with probability $q=1-(1-p)^{1/\NumOfLayers}$ and then sample each
real node $v\in V$ if and only if at least one of its virtual copies
is sampled (i.e., the sampled real nodes are obtained by projecting
the sampled virtual nodes onto $G$). The probability of each real node
being sampled is $1-(1-q)^{\NumOfLayers} = p$. Henceforth, we work on
$\mathcal{G}$ assuming that each virtual node is sampled independently
with probability $q = 1-(1-p)^{1/\NumOfLayers} \geq
\frac{p}{2\NumOfLayers}$. Let $\mathcal{V}_\ell$ be the set of all
virtual nodes in layers $1,\dots,\ell$ and let $\mathcal{S}_\ell$ be
the set of all sampled virtual nodes from $\mathcal{V}_\ell$. We
define $N_\ell$ to be the number of connected components of
$\mathcal{G}[\mathcal{S}_\ell]$ and let $M_{\ell}= N_{\ell}-1$ (i.e.,
the excess number of components).

\paragraph{Proof Outline :} We first show that the sampled virtual
nodes of the first $\NumOfLayers/2$ layers already give
domination. Then, for each layer $\ell \geq \NumOfLayers/2$, we look
at components of $\mathcal{G}[\mathcal{S}_{\ell}]$ and show that
adding sampled virtual nodes of layer $\ell+1$ merges (in expectation)
a constant fraction of these components, each with at least one other
component, while not creating new components. Consequently, we get
that $M_\ell$ decreases essentially exponentially with $\ell$, until
it becomes zero, at which point connectivity is attained.  Formally,
we use the following two key lemmas.

\begin{lemma}[\textbf{Domination Lemma}] \label{lem:dom} With high
  probability, $\mathcal{S}_{\NumOfLayers/2}$ dominates
  $\mathcal{V}$.
\end{lemma}
The proof of the Domination Lemma is straightforward and thus deferred
to \Cref{app7}. In the following, let $\mathcal{D}$ be the set of
dominating sets of $\mathcal{G}$ consisting only of nodes from layers
$1,\dots,\NumOfLayers/2$. Further, for all $\ell\geq \NumOfLayers/2$,
let $\mathcal{D}_{\ell}$ contain all sets $\mathcal{T}\subseteq
\mathcal{V}_{\ell}$ such that there exists a $D\in \mathcal{D}$ for
which $D\subseteq \mathcal{T}$. That is, $\mathcal{D}_{\ell}$ is the collection of all sets of virtual nodes in layers $1,\dots,\ell$ that contain a dominating set $D \in \mathcal{D}$. Note that using this notation, Lemma
\ref{lem:dom} states that $S_{\NumOfLayers/2}\in \mathcal{D}$, w.h.p.

\begin{lemma} [\textbf{Fast Merger Lemma}]\label{lem:FML} If
  $\mathcal{S}_{\NumOfLayers/2}\in \mathcal{D}$, then for each
  ${\ell} \geq \frac{\NumOfLayers}{2}$, we have $M_{\ell+1} \leq
  M_\ell$. Moreover, for every ${\ell} \geq \frac{\NumOfLayers}{2}$
  and every $\mathcal{T}\in \mathcal{D}_\ell$, we have $\Pr\big[M_{{\ell}+1}
  \leq \frac{9}{10} \cdot M_{\ell} \,|\,
  \mathcal{S}_\ell=\mathcal{T}\big]\geq 1/3$.
  % dominates $\mathcal{V}$ and independent of any other let
  % $\mathcal{A}_{\ell}$ be the event that $M_{{\ell}+1} \leq
  % \frac{7}{8} \cdot M_{\ell}$. Then, $\Pr[\mathcal{A}_{\ell}] \geq
  % 1/3$ and the events $\mathcal{A}_\ell$ for different layers ${\ell}
  % \geq \frac{\NumOfLayers}{2}$ are independent.
\end{lemma}

Thus, as long as $\mathcal{S}_{\NumOfLayers/2}$ is a dominating set,
the probability for reducing the number of components by a constant
factor is at least $1/3$, independently of the result of the sampling
in layers $1,\dots,\ell$.

\paragraph{Connector Paths:} To prove the Fast Merger Lemma, we define
\emph{connector paths}. Consider the transition from
$\mathcal{G}[\mathcal{S}_\ell]$ to
$\mathcal{G}[\mathcal{S}_{\ell+1}]$, suppose $N_\ell \geq 2$, and let
$\mathcal{C}$ be a connected component of
$\mathcal{G}[\mathcal{S}_{\ell}]$. Consider the projection
$\Psi(\mathcal{S}_{\ell})$ onto
$G$. %, and its $\ell+1$ copy in $\mathcal{G}$.
We say that a path $P$ of $G$ is a \emph{potential connector path} for
component $\Psi(\mathcal{C})$ if the following two conditions hold:
(A) $P$ has one endpoint in $\Psi(\mathcal{C})$ and another endpoint
in $\Psi(\mathcal{S}_{\ell} \setminus \mathcal{C})$ and (B) $P$ has at
most $2$ internal nodes. From $P$ we derive a potential connector path
$\mathcal{P}$ of $\mathcal{C}$ in $\mathcal{G}$ by taking the layer
$\ell+1$ copy in $\mathcal{G}$ of the internal nodes of $P$, and for
each endpoint $w$ of $P$ we take its copy in $\mathcal{S_{\ell}}$.  We
say that a potential connector path $\mathcal{P}$ is a \emph{connector
  path} if and only if in the layer $\ell+1$ copy of $G$, the internal
(virtual) nodes of $\mathcal{P}$ are all sampled. Note that if
$\mathcal{C}$ has a connector path to another component
$\mathcal{C}'$, then in $\mathcal{G}[\mathcal{S}_{\ell+1}]$, the
components $\mathcal{C}$ and $\mathcal{C}'$ are merged with each
other. Given this, in order to prove the Fast Merger Lemma, we show
\Cref{lem:paths} and prove it. Given \Cref{lem:paths}, we can prove
\Cref{lem:FML} and consequently \Cref{thm:threshold} using standard
probability arguments. The details of those parts are deferred to
\Cref{app7}.

\begin{lemma} \label{lem:paths} For each $\ell\in [{\NumOfLayers}/{2}, \NumOfLayers-1]$, and every $\mathcal{T}\in \mathcal{D}_\ell$, if $\mathcal{S}_\ell=\mathcal{T}$ and $N_\ell \geq 2$, then for each connected component
  $\mathcal{C}$ of $\mathcal{G}[\mathcal{S}_{\ell}]$, with probability at least $1/2$, $\mathcal{C}$ has at least one connector path.
\end{lemma}

\begin{proof}[Proof of \Cref{lem:paths}]
	Fix a layer $\ell\in [{\NumOfLayers}/{2}, \NumOfLayers-1]$, and an arbitrary set $\mathcal{T}\in \mathcal{D}_\ell$ and fix $\mathcal{S}_\ell=\mathcal{T}$.
  Consider the projection $\Psi(\mathcal{S}_{\ell})$ onto $G$ and
  recall Menger's theorem: Between any pair $(u,v)$ of non-adjacent
  nodes of a $k$-vertex connected graph, there are $k$ internally
  vertex-disjoint paths connecting $u$ and $v$. Applying Menger's
  theorem to a node in $\Psi(\mathcal{C})$ and a node in
  $\Psi(\mathcal{S}_\ell\setminus\mathcal{C})$, we obtain at least $k$
  internally vertex-disjoint paths between $\Psi(\mathcal{C})$ and
  $\Psi(\mathcal{S}_\ell \setminus \mathcal{C})$ in $G$. We first show
  that these paths can be shortened so that each of them has at most
  $2$ internal nodes i.e., to get property (B) of potential connector
  paths. Pick an arbitrary one of these $k$ paths and denote it $P$ =
  $v_1$, $v_2$, $...$, $v_r$, where $v_1 \in \Psi(\mathcal{C})$ and
  $v_r \in \Psi(\mathcal{S}_\ell \setminus \mathcal{C})$. By the
  assumption that $\mathcal{S}_{\NumOfLayers/2}$ dominates
  $\mathcal{G}$, we get that $\Psi(\mathcal{S}_\ell)$ dominates
  $G$. Hence, since $v_1 \in \Psi(\mathcal{C})$ and $v_r \in
  \Psi(\mathcal{S}_\ell \setminus \mathcal{C})$, either there is a
  node $v_i$ along $P$ that is connected to both $\Psi(\mathcal{C})$
  and $\Psi(\mathcal{S}_\ell \setminus \mathcal{C})$, or there must
  exist two consecutive nodes $v_i$, $v_{i+1}$ along $P$, such that
  one of them is connected to $\Psi(\mathcal{C})$ and the other is
  connected to $\Psi(\mathcal{S}_\ell \setminus \mathcal{C})$. In
  either case, we can derive a new path $P'$ which satisfies (B) and
  is internally vertex-disjoint from the other $k-1$ paths since its
  internal nodes are a subset of the internal nodes of $P$ and are not
  in $\Psi(\mathcal{S}_{\ell})$. After shortening all the $k$
  internally vertex-disjoint paths, we get $k$ internally
  vertex-disjoint paths satisfying (A) and (B), i.e., $k$ internally
  vertex-disjoint \emph{potential connectors paths} for
  $\Psi(\mathcal{C})$.

  For each potential connector path $P$ of these $k$ vertex-disjoint
  potential connector paths, each virtual internal node of its copy
  $\mathcal{P}$ in layer $\ell+1$ of $G$ is sampled with probability
  $q$. Since each potential connector path $P$ has at most $2$
  internal nodes, the probability that all internal virtual nodes of
  its copy $\mathcal{P}$ in layer $\ell+1$ are sampled is at least
  $q^2=(\frac{p}{2\NumOfLayers})^2 \geq (\frac{\alpha}{2\lambda})^2
  \cdot \frac{1}{k}$. Hence, the expected number of connector paths is
  at least $(\frac{\alpha}{\lambda})^2$. Because the $k$ potential
  connector paths are internally vertex-disjoint, the events of their
  copies being connector paths (having their internal virtual nodes in
  layer $\ell+1$ sampled) are independent. Therefore, choosing a large
  enough constant $\alpha$, we get that with probability at least
  $1/2$, component $\mathcal{C}$ has at least one connector path.
\end{proof}

%%% Local Variables: 
%%% mode: latex
%%% TeX-master: "main_FOCS"
%%% End: 

\section{Construction of CDS Packing, and Vertex Connectivity After Sampling}\label{sec:general}

In this section, we present our main CDS packing result. 

\begin{theoremR}{thm:packing} Any graph with vertex connectivity
  $k\geq 1$ has a CDS packing of size $\Omega(\frac{k}{\log n}+1)$.
\end{theoremR}

We complement the above lower bound by an upper bound on CDS packing,
showing that \Cref{thm:packing} is existentially optimal and loosing
the $\Theta(\log n)$ factor is unavoidable. The proof of the following
\Cref{thm:BadExample}, which is deferred to \Cref{sec:BadExample},
builds on a base graph $\mathsf{H}$, which has vertex connectivity
$O(\log n)$ and maximum CDS packing size $O(1)$. It then uses the
probabilistic method\cite{ProbMethod92} to show that for any value of
$k$, there exists a subgraph of $\mathsf{H}$ which shows the claimed
logarithmic gap.

\begin{theoremRF} {thm:BadExample} For any sufficiently large $n$ and
  any $k \in [1, n/4]$, there exists a $k$-vertex-connected, $n$-node
  graph $G$ such that the maximum CDS packing of $G$ is of size
  $O\big(\frac{k}{\log n}+1\big)$.
\end{theoremRF}

Our proof of \Cref{thm:packing} is constructive and it also shows that
a CDS packing of size $O(k/\log n)$ can be found efficiently. Even
more generally, the same construction allows to create a CDS packing
of $G$ of size $\Omega({k p^2}/{\log^3 n})$, if nodes are
independently sampled with probability $p$ and the CDSs of the packing
are restricted to only consist of sampled nodes. Using
\Cref{prop:upp}, this CDS packing acts as a witness to the vertex
connectivity of the sampled graph and it thus proves our main sampling
theorem.
\begin{theoremR}{thm:general}
  Consider an arbitrary graph $G=(V,E)$ with vertex connectivity
  $k$. Let $S$ be a subset of $V$ where each node $v\in V$ is included
  in $S$ (i.e., \emph{sampled}) independently with probability
  $p$. W.h.p., the graph $G[S]$ has a CDS packing of size
  $\Omega(\frac{k p^2}{\log^3 n})$, and thus, also vertex connectivity
  $\Omega(\frac{k p^2}{\log^3 n})$.
\end{theoremR}

To capture Theorems \ref{thm:packing} and \ref{thm:general} together,
in the remainder of this section, we discuss the construction in the
case of vertex sampling with probability $p$ and give a CDS packing
with size $\Omega(kq^2/ \log n)$, where
$q=1-(1-p)^{1/(3\NumOfLayers)}$ and $L=\lambda\log n=\Theta(\log n)$
is the number of layers (as in \Cref{sec:threshold}). Since $q \geq
\frac{p}{6\NumOfLayers}=\Theta(\frac{p}{\log n})$, the CDS packing
size is $\Omega(kq^2/\log n) = \Omega({k p^2}/{\log^3 n})$, thus
proving \Cref{thm:general}. When we set $p=1$ (i.e., no sampling), we
get $q=1$ and thus, the CDS packing size becomes $\Omega({k}/{\log
  n})$ as claimed by \Cref{thm:packing}.

\subsection{Construction of the CDS Packing}\label{subsec:partVirtual}
To construct the claimed CDS packing, similarly to
\Cref{sec:threshold}, we transform the graph $G=(V,E)$ into a virtual
graph $\mathcal{G}=(\mathcal{V}, \mathcal{E})$, this time consisting
of $3\NumOfLayers$ copies of $G$. As in \Cref{sec:threshold}, to
translate the sampling to $\mathcal{G}$, we sample each virtual node
with probability $q = 1-(1-p)^{1/(3\NumOfLayers)} \geq
\frac{p}{6\NumOfLayers}$. To construct the promised CDS packing on the
sampled real nodes, we create a CDS \emph{partition} of size
$\Omega(kq^2)=\Omega({k p^2}/{\log^2 n})$ on the sampled virtual
nodes. Since each real node has $\Theta(\log n)$ virtual copies, by
giving a weight of $1/\Theta(\log n)$ to each CDS, we directly get the
claimed CDS packing.

In the rest of this section, we work on $\mathcal{G}$ and show how to
construct $t = \delta\cdot\, kq^2$ vertex-disjoint connected
dominating sets on the sampled virtual nodes, for a sufficiently small
constant $\delta>0$. We have $t$ \emph{classes} and we assign each
sampled virtual node to one class, such that, eventually each class is
a CDS, w.h.p. To organize the construction, we group the virtual nodes
in $\NumOfLayers$ layers, putting three copies of graph $G$ in each
layer. Inside each layer, the three copies are distinguished by a
\emph{type number} in $\set{1,2, 3}$.

We use notations (partially) similar to \Cref{sec:threshold}. Let
$\mathcal{S}^i_\ell$ be the set of sampled nodes of layers $1$ to
$\ell$ that are assigned to class $i$. Let $N^i_\ell$ be the number of
connected components of $\mathcal{G}[\mathcal{S}^i_\ell]$. Finally,
define $M_\ell:=\sum_{i=1}^t (N_\ell^i-1)$ to be the total number of
excess components after considering layers $1,\dots,\ell$. Initially
$M_1 \leq n-t$, and as soon as $M_\ell=0$, each class induces a
connected sub-graph. Thus, the goal will be to assign sampled virtual
nodes to classes such that $M_\NumOfLayers=0$.

The class assignments are performed in a recursive manner based on the layer numbers. We begin the assignment with a jump-start, assigning sampled virtual nodes of layers $1$ to $\frac{\NumOfLayers}{2}$ to random classes. We show in ~\Cref{lem:layer1dom} that this already gives domination. The proof is deferred to \Cref{app8}.
\begin{lemma}[\textbf{Domination Lemma}] \label{lem:layer1dom} 
  W.h.p., for each class $i$, $\mathcal{S}^i_{\NumOfLayers/2}$ dominates $\mathcal{V}$.
\end{lemma}

Note that the domination of each class follows directly from this lemma. For the rest of this section, we assume that for each class $i$, $\mathcal{S}^i_{\NumOfLayers/2}$ dominates $\mathcal{V}$, and we will use this property to get short connector paths. 

After the first $\frac{\NumOfLayers}{2}$ layers, we go over the layers one by one and for each layer $\ell \in [{\NumOfLayers}/{2}, \NumOfLayers-1]$, we assign nodes of layer $\ell+1$ to classes based on the assignments of nodes of layers $1$ to $\ell$. In the rest of this section, we explain this assignment for layer $\ell+1$. We refer to nodes of layers $1$ to $\ell$ as \emph{old nodes} whereas nodes of layer $\ell+1$ are called \emph{new nodes}. The goal is to perform the class assignment of the new sampled nodes such that (in expectation) the number of connected components decreases by a constant factor in each layer, i.e., to get a Fast Merger Lemma. During the recursive assignments, our main tool will be a modified variant of the concept of \emph{connector paths}, which we define next.

\subsection{Connector Paths}\label{subsec:Connectors}
Consider a class $i$, suppose $N^i_{\ell} \geq 2$, and consider a
component $\mathcal{C}$ of $\mathcal{G}[\mathcal{S}^i_\ell]$. We use
the projection $\Psi(\mathcal{S}^i_\ell)$ onto graph $G$ as defined in
\Cref{sec:threshold}. A path $P$ in $G$ is called a \emph{potential
  connector} for $\Psi(\mathcal{C})$ if it satisfies the following
three conditions: (A) $P$ has one endpoint in $\Psi(\mathcal{C})$ and
the other endpoint in $\Psi(\mathcal{S}^i_{\ell} \setminus
\mathcal{C})$, (B) $P$ has at most two internal nodes, (C) if $P$ has
exactly two internal nodes and has the form $s$, $u$, $w$, $t$ where
$s\in \Psi(\mathcal{C})$ and $t \in \Psi(\mathcal{S}^i_{\ell}
\setminus \mathcal{C})$, then $w$ does not have a neighbor in
$\Psi(\mathcal{C})$ and $u$ does not have a neighbor in
$\Psi(\mathcal{S}^i_{\ell} \setminus \mathcal{C})$. Condition (C) is
an important condition, which is new compared to the definition in
\Cref{sec:threshold}.\footnote{We could add condition (C) to the
  definition of potential connector paths in \Cref{sec:threshold} as
  well, but it is not necessary there.} The condition requires that
each potential connector path is minimal, i.e., there is no potential
connector of length $2$ connecting $\Psi(\mathcal{C})$ to another
component of $\Psi(\mathcal{S}^i_{\ell})$ via only $u$ or only $w$.

From a potential connector path $P$ on graph $G$, we derive a
potential connector path $\mathcal{P}$ on virtual nodes in
$\mathcal{G}$ by determining the types of related internal nodes as
follows: (D) If $P$ has one internal real node $w$, then for
$\mathcal{P}$ we choose the virtual node of $w$ in layer ${\ell}+1$ in
$\mathcal{G}$ with type $1$. (E) If $P$ has two internal real nodes
$w_1$ and $w_2$, where $w_1$ is adjacent to $\Psi(\mathcal {C})$ and
$w_2$ is adjacent to $\Psi(\mathcal{S}^i_{\ell} \setminus \mathcal
{C})$, then for $\mathcal{P}$ we choose the virtual nodes of $w_1$ and
$w_2$ in layer $\ell+1$ with types $2$ and $3$, respectively. Finally,
for each endpoint $w$ of $P$ we add  the copy of $w$ in
$\mathcal{S}^i_{\ell}$ to $\mathcal{P}$. 

A given potential connector path $\mathcal{P}$ on the virtual nodes of
layer $\ell+1$ is called a \emph{connector path} if and only if the
internal nodes of $\mathcal{P}$ are sampled. We call a connector path
that has one internal node a \emph{short connector path}, whereas a
connector path with two internal nodes is called a \emph{long
  connector path}.

Because of condition (C), and rules (D) and (E) above, we get the
following important fact:
\begin{proposition}\label{prop:longpaths}
  For each class $i$, each type-$2$ virtual node $u$ of layer $\ell+1$
  is on connector paths of at most one connected component of
  $\mathcal{G}[\mathcal{S}^i_{\ell}]$.
\end{proposition}

\begin{figure}[t]
\centering
\includegraphics[width=0.7\textwidth]{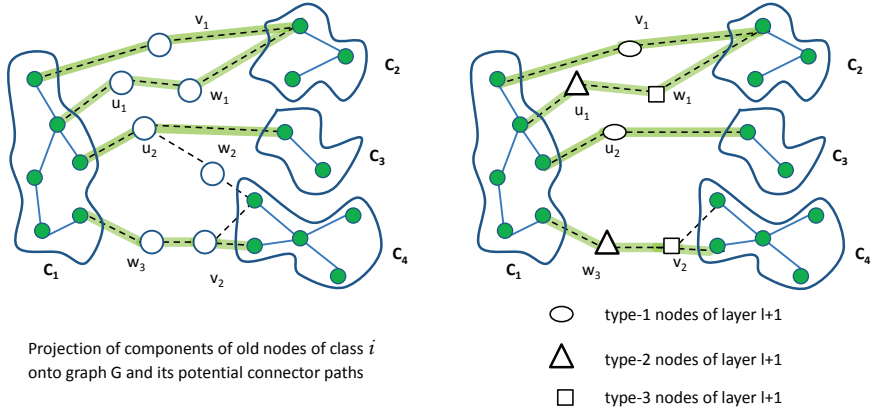}
	\caption{Potential Connector Paths for component $\mathcal{C}_1$ in layer $\ell+1$ copies of $G$}
	\label{fig:ConnectorPaths}
	\vspace{-0.2cm}
\end{figure}

\Cref{fig:ConnectorPaths} demonstrates an example of potential
connector paths for a component $\mathcal{C}_1 \in
\mathcal{G}[\mathcal{S}^i_\ell]$. The figure on the left shows graph
$G$, where the projection $\Psi(\mathcal{S}^i_\ell)$ is indicated via
green nodes, and the green paths are potential connector paths of
$\Psi(\mathcal{C}_1)$. On the right side, we see the same potential
connector paths, where the type of the related internal nodes are
determined according to rules (D) and (E) above, and nodes of
different types are distinguished via different shapes (for clarity,
virtual nodes of other types are omitted from the figure).
 
%The top path connecting $\mathcal{C}_1$ with $\mathcal{C}_2$ through type-$1$ node $v_1$ is a short connector path for $\mathcal{C}_1$. The second path connecting $\mathcal{C}_1$ with $\mathcal{C}_2$ through type-$2$ node $u_1$ and type-$3$-node $w_1$ is a long connector path for $\mathcal{C}_1$. Moreover, note that the path which goes through nodes $u_2$ and $w_2$ is not a connector path because it does not satisfy ($Q_3$) as there is a shorter path connecting $Q_3$ through node $u_2$. Note that this shorter path is not a connector path itself because it does not satisfy type restriction of ($Q'_1$). Note that in the above example, if we did not have condition ($Q_3$), the path segment with internal nodes $u_2$ and $w_2$ which goes to $\mathcal{C}_4$ would be a long connector path for both $\mathcal{C}_1$ and $\mathcal{C}_3$, violating \Cref{prop:longpaths} for $u_2$. We remark that \Cref{prop:longpaths} will play a key role in our analysis and it is the main reason for having types. 

The following lemma shows that each component that is not alone in its
class is likely to have many connector paths. The proof is similar to
that of \Cref{lem:paths}, and is deferred to \Cref{app8}.

\begin{lemma} [\textbf{Connector Abundance Lemma}] \label{lem:CAL}
  Consider a layer $\ell \geq \NumOfLayers/2$ and a class $i$ such
  that $\mathcal{S}_{L/2}^i\subseteq \mathcal{S}_\ell^i$ is a
  dominating set of $\mathcal{G}$ and $N^i_{\ell}\geq 2$. Further
  consider an arbitrary connected component $\mathcal{C}$ of
  $\mathcal{G}[\mathcal{S}^i_\ell]$. Then, with probability at least
  $1/2$, $\mathcal{C}$ has at least $t$ internally vertex-disjoint
  connector paths.
\end{lemma}

%Suppose that we have identified the connected components of old nodes in all classes. Then for each connected component that is not alone in its class, identify $k$ potential connector paths on real nodes. From there, we get $k$ potential connector paths on virtual nodes which have internal nodes in layer $\ell+1$ and have the right types. Depending on whether the internal nodes of these paths are sampled or not, we get a number of connector paths for each component. 

\subsection{Recursive Step of Class Assignment} \label{subsec:recurse}
From \Cref{lem:CAL} we know that for each connected component of each
class $i$ with $N^i_\ell \geq 1$, with probability at least $1/2$,
this component has at least $t$ connector paths. For each such
component, pick exactly $t$ of its connector paths. Using Markov's
inequality, we get that with probability at least $1/4$, we have at
least $M_\ell \cdot t$ connector paths in total, over all the classes
and components. We use these connector paths to assign the class
numbers of nodes of layer $\ell+1$. This part is done in a greedy
fashion, in three stages as follows:

\begin{itemize}
\item[(I)] For each type-$1$ new node $v$: For each class $i$, define
  the \emph{class-$i$ degree of $v$} to be the number of connected
  components of class $i$ which have a \emph{short} connector path
  through $v$. Let $\Delta$ be the maximum class-$i$ degree of $v$ as
  $i$ ranges over all classes, and let $i^*$ be a class that attains
  this maximum. If $\Delta \geq 1$: Assign node $v$ to class
  $i^*$. Also, remove all connector paths of all classes that go
  through $v$. Moreover, remove all connector paths of all the
  connected components of class $i^*$ which have $v$ on their short
  connector paths.

\item[(II)] For each type-$3$ new node $u$: For each class $i$, define
  the \emph{class-$i$ degree of $u$} to be the number of connected
  components of class $i$ which have a \emph{long} connector path
  through $u$. Let $\Delta$ be the maximum class-$i$ degree of $u$ as
  $i$ ranges over all classes and let $i^*$ be a class that attains
  this maximum. If $\Delta \geq 1$: Assign $u$ to class
  $i^*$. Moreover, each of the $\Delta$ long connector paths of class
  $i^*$ that passes through $u$ also has a type-$2$ internal node.
  Let these type-$2$ nodes be $v_1,\dots,v_{\Delta}$. We assign the
  nodes $v_1,\dots,v_\Delta$ to class $i^*$. Then, all the connector
  paths that go through $u$ or any of the nodes $v_1, \dots,
  v_{\Delta}$ are removed. We also remove all connector paths of each
  component of class $i^*$ that has a connector path going through
  $u$.

\item[(III)] Assign each remaining new node to a random class.

\end{itemize}

%-----------------------------

\begin{lemma} [\textbf{Fast Merger Lemma}] \label{lem:FML2} For each
  ${\ell} \geq \frac{\NumOfLayers}{2}$ and every assignment of the
  sampled layer $1,\dots,\ell$ nodes to classes such that for all
  classes $i$, $\mathcal{S}_{\NumOfLayers/2}^i$ is a dominating set of
  $\mathcal{G}$, with probability at least $1/4$, $M_{{\ell}+1} \leq
  \frac{5}{6} \cdot M_{\ell}$.
\end{lemma}
The proof is based on an accounting argument that uses the total
number of remaining connector paths over all classes and components as
the budget, and shows that number of components that are merged, each
with at least one other component, is large. We defer the details to
\Cref{app8}.

%%% Local Variables: 
%%% mode: latex
%%% TeX-master: "main_FOCS"
%%% End: 

\section{Construction of a CDS Partition}
\label{sec:partition}

\begin{theoremRF}{thm:partition}
  Every $k$-vertex-connecgted graph $G$ has a CDS partition of size
  $\Omega\big(\frac{k}{\log^5 n}\big)$, and if $k =\Omega(\sqrt{n})$,
  then $G$ has a CDS partition of size $\Omega\big(\frac{k}{\log^2
    n}\big)$. Moreover, such CDS partitions can be found in polynomial
  time.
\end{theoremRF}

In order to achieve this CDS partition, we use an algorithm in the
style of the CDS packing of \Cref{thm:general}. Here, we explain the
key changes: since in a CDS partition, each node can only join one
CDS, we cannot use the layering style of \Cref{thm:general}, which
uses $\Theta(\log n)$ copies of $G$ and where each node can join
$O(\log n)$ CDSs. Instead, we use \emph{random} layering: each
node chooses a random \emph{layer number} in $\set{1, \dots, L}$ and a
random \emph{type number} in $\set{1,2,3}$. The construction is again
recursive, with first assigning nodes of layer $1$ randomly to one of
$t$ random classes. This suffices to give domination. Here, $t =\delta
\frac{k}{\log^2 n}$ if $k =\Omega(\sqrt{n})$ and $t=\delta
\frac{k}{\log^5 n}$, otherwise. After that, for each $\ell \geq 2$, we
assign class numbers of nodes of layer $\ell+1$ based on the
configuration of classes in layers $1$ to $\ell$, using the same
greedy algorithm as in \Cref{subsec:recurse}. Next, we re-define the
connector paths, incorporating the random layering.

\paragraph{Connector Paths for CDS Partition:} Let $V^i_\ell$ be the
set of all nodes of layers $1$ to $\ell$ in class $i$. Consider a
component $\mathcal{C}$ of $G[V^i_\ell]$. Define potential connector
paths on $G$ as in \Cref{subsec:Connectors} (conditions (A) to
(C)). Then, for each potential connector path on $G$, this path is
called a \emph{connector path} if its internal nodes are in layer
$\ell+1$ and the types of its internal nodes satisfy rules (D) and (E)
in \Cref{subsec:Connectors}.

The main technical change, with respect to the CDS packing of
\Cref{sec:general}, appears in obtaining a Connector Abundance Lemma,
which we present in two versions, depending on the magnitude of
vertex-connectivity.

\begin{lemma} [\textbf{Connector Abundance
    Lemma}] \label{lem:CAL-Part} For each class $i$ and layer ${\ell}
  \leq {\NumOfLayers}/{2}$ such that $N^i_{\ell}\geq 2$, for each
  connected component $\mathcal{C}$ of $G[V^i_{\ell}]$, with
  probability at least $1/2$, $\mathcal{C}$ has at least
  $\Omega\big(\frac{k}{\log^5 n}\big)$ internally vertex-disjoint
  connector paths, with independence between different layers ${\ell}
  \leq {\NumOfLayers}/{2}$.
\end{lemma}

\begin{lemma} [\textbf{Stronger Connector Abundance Lemma for Large
    Vertex Connectivity}] \label{lem:CAL-Part-Large} Assume $k =
  \Omega(\sqrt{n})$. W.h.p., for each class $i$ and layer ${\ell} \leq
  {\NumOfLayers}/{2}$ such that $N^i_{\ell}\geq 2$, each connected
  component $\mathcal{C}$ of $G[V^i_{\ell}]$ has at least
  $\Omega\big(\frac{k}{\log^2 n}\big)$ internally vertex-disjoint
  connector paths.
\end{lemma}

The proof of \Cref{lem:CAL-Part} is relatively straightforward
application of \Cref{thm:general}.
\begin{proof}[Proof of \Cref{{lem:CAL-Part}}]
  Let $W^*_\ell$ be the set of all nodes with a layer number in
  $\{\ell+1,\dots, \NumOfLayers\}$. Since the probability of each node
  to be in $V^*_\ell$ is at least $1/2$ (because ${\ell} \leq
  {\NumOfLayers}/{2}$), \Cref{thm:general} shows that, w.h.p, the
  vertex-connectivity of $G[W^*_\ell]$ is $\Omega(\frac{k}{\log^3
    n})$. It is easy to see that therefore, the vertex-connectivity of
  $G[W^*_\ell \cup S^i_{\ell}]$ is also $\Omega(\frac{k}{\log^3
    n})$. Thus, for each component $\mathcal{C}$ of $G[V^i_{\ell}]$,
  we can follow the proof of \Cref{lem:paths}, this time using
  Menger's theorem on $G[W^*_\ell \cup S^i_{\ell}]$, and find
  $\Omega(\frac{k}{\log^3 n})$ internally vertex-disjoint potential
  connector paths in graph $G[W^*_\ell \cup V^i_{\ell}]$. It is clear
  that the internal nodes of these potential connector paths are not
  in $V^i_{\ell}$, which means they are in $W^*_\ell$. For each
  potential connector path, for each of its internal nodes, given that
  this node is in $W^*_\ell$, the probability that the node is in
  layer $\ell+1$ and has the type which satisfies rules (A) and (B) of
  \Cref{subsec:Connectors} is at least $\Theta(1/{\NumOfLayers}) =
  \Theta(1/\log n)$. Hence, the probability of each of these potential
  connector paths being a connector path is at least $\Theta(1/\log^2
  n)$. From internally vertex-disjointedness of the potential
  connector paths, and since there are $\Omega(\frac{k}{\log^3 n})$ of
  them, it follows that with probability at least $1/2$, component
  $\mathcal{C}$ has at least $\Omega(\frac{k}{\log^5 n})$ internally
  vertex-disjoint connector paths.
\end{proof}

The proof of \Cref{lem:CAL-Part-Large} is more involved. Intuitively,
for each class $i$, it first contracts components of $G[V^i_1]$ and
then argues about all $2^{O(n/k)}$ cuts of the resulting graph, using
the fact that the large vertex connectivity $k = \Omega(\sqrt{n})$
gives a good enough concentration to compensate for this large number
of cuts. This proof is deferred to \Cref{app9}.

%%% Local Variables: 
%%% mode: latex
%%% TeX-master: "main_FOCS"
%%% End: 

\section{Concluding Remarks} 
\label{sec:discussions}

In this paper, we introduce CDS partition and CDS packing as a new
way of looking at the structure of graphs with a given vertex
connectivity $k$. 
% The problem of CDS partition is to partition the
% vertices of the graph into as many connected dominating sets as
% possible. CDS packing is the related fractional relaxation, which
% allows different CDSs to share nodes but normalizes CDSs by a giving a
% weight to each one such that the total weight in each node is bounded
% by $1$. 
We argue that CDS packing and CDS partition can be viewed as
the ``vertex world" counterparts of the well-known edge-disjoint
spanning trees (often attributed to 1961 works of
Nash-Williams\cite{Nash-Williams} and Tutte\cite{Tutte}). We show
that each $k$-vertex-connected graph with $n$ nodes has a CDS packing
of size $\Omega(k / \log n)$ and that this is existentially
optimal. We also show that any such graph has a CDS partition of
size $\Omega(k/\log^5 n)$ and if $k=\Omega(\sqrt{n})$, it has a CDS
partition of size $\Omega(k/\log^2 n)$.

Using the new perspective provided by CDS packing, we manage to break
the barrier of facing exponentially many small vertex cuts and prove
that after sampling vertices with probability $p$, the remaining
vertex connectivity is $\tilde{\Omega}(kp^2)$, w.h.p. We also explain
that CDS packing is useful in networking as it provides a backbone for
store-and-forward algorithms with optimal throughput. In a subsequent
work \cite{Distributed-CDS-Packing}, we show that a CDS packing of
size $O(k/\log n)$ can also be computed efficiently by a distributed
algorithm with time complexity $\tilde{O}(n/k)$.

Our work leaves a number of interesting and important open
questions. First of all, most of our results are only tight up to
logarithmic factors and it would certainly be interesting to try to
close these gaps. When sampling each node with probability $p$,
assuming that $kp^2$ is large enough, we suspect that the remaining
vertex connectivity should be $\Theta(kp^2)$. We also believe that
construction from Section \ref{sec:partition} to compute a CDS
partition should actually allow to obtain a CDS partition of size
$\Omega(k/\log^2n)$ for all values of $k$ and not just for
$k=\Omega(\sqrt{n})$. Note that our approach would immediately give
this if we could prove that when sampling with probability $1/2$, the
remaining vertex connectivity is still $\Omega(k)$ rather than
$\Omega(k/\log^3 n)$ as we show. Note that even proving an
$\Omega(k/\log^2 n)$ lower bound on the size of a maximum CDS
partition would still leave a logarithmic gap w.r.t.\ the upper bound
of Theorem \ref{thm:BadExample}. Finally, we hope that CDS packings
and partitions can be used to approach other problems related to
vertex connectivity. In particular, if true, the approach might help
to show that the sizes all vertex cuts remain around their
expectations after vertex sampling. We believe that such a result
would be a major contribution.

%%% Local Variables: 
%%% mode: latex
%%% TeX-master: "main_FOCS"
%%% End: 

\section{Acknowledgment}
We are grateful to David Karger for helpful discussions and we also
thank Bernhard Haeupler for valuable comments.

{\small
\bibliographystyle{abbrv}
\bibliography{ref}
}

\appendix
\section{Upper Bounding the Size of CDS-Packing}\label{sec:BadExample}
In this section, we first present the missing proof of \Cref{prop:upp}, and then, in the more interesting part, prove \Cref{thm:BadExample}, which shows that in some graphs, the maximum CDS packing size is a $\Theta(\log n)$ factor smaller than the vertex connectivity.

\begin{proof}[Proof of \Cref{prop:upp}] Consider a vertex cut $\mathcal{C} \subseteq V$ of $G$
    that has size exactly $k$. Each CDS $\tau$ must include at least
    one node in $\mathcal{C}$. For each CDS $\tau\in \CDS(G)$, pick one
    node $v \in C$ as a representative of $\tau$ in the cut and let us
    denote it by $\mathit{Rep}(\tau)$. Thus, for any CDS-Packing of $G$, we
    have
    $$\sum_{\tau \in \CDS(G)} x_{\tau} = \sum_{v\in \mathcal{C}} \;\;\; \sum_{\tau \in \CDS(G) \atop s.t.\; v = \mathit{Rep}(\tau)} x_{\tau} \leq \sum_{v\in \mathcal{C}} 1= |\mathcal{C}| =k. $$
    Since the above holds for any CDS-Packing of $G$, we get that
    $K'_{\CDS}(G) \leq k$.
\end{proof}

\medskip
Let us recall the statement of \Cref{thm:BadExample}
\begin{theoremR}{thm:BadExample} For any large enough $n$ and any $k\geq 1$,
  there exists an $n$-node graph $G$ with vertex connectivity $k$ such
  that $K'_{\CDS}(G) =O\big(\frac{k}{\log n}+1\big)$.
\end{theoremR}

To prove this theorem, in \Cref{lem:badgraph}, we present a graph $\mathsf{H}$ with vertex connectivity $k$,
size between $2^k$ and $4^k$, and $K'_{\CDS}(\mathsf{H}) < 2$. This lemma
proves the theorem for $k=O(\log n)$. To prove the theorem for the case of
larger vertex-connectivity compared to $n$, in \Cref{lem:badgraph_larger_k}, we look at randomly chosen sub-graphs of
$\mathsf{H}$ and apply the probabilistic method~\cite{ProbMethod92}.

\begin{lemma} \label{lem:badgraph}For any $k$, there exists an
  $n$-node graph $\mathsf{H}$ with vertex connectivity $k$ and
  $n\in[2^k, 4^k]$ such that $K'_{\CDS}(\mathsf{H}) < 2$.
\end{lemma}
\begin{proof} 
  We obtain graph $\mathsf{H}$ by simple modifications to the graph
  presented by Sanders et~al.~\cite{sandersPolynomial} for proving an
  $\Omega(\log n)$ network coding gap in the model where each node can
  send distinct unit-size messages to its different neighbors.

  The graph $\mathsf{H}$ has two layers. The first layer is a clique
  of $2k$ nodes. The second layer has $\binom{2k}{k}$ nodes, one for
  each subset of size $k$ of the nodes of the first layer. Each second
  layer node is connected to the $k$ first-layer nodes of the
  corresponding subset. Note that the total number of nodes is
  $\binom{2k}{k}+2k \in [2^k, 4^k]$.  Let $\mathcal{A}$ and
  $\mathcal{B}$ denote the set of nodes in the first and
  second layer, respectively.

  First, we show that $\mathsf{H}$ has vertex connectivity $k$. Since
  the degree of each second layer node is exactly $k$, it is clear
  that the vertex connectivity of $\mathsf{H}$ is at most $k$. To
  prove that the vertex connectivity of $\mathsf{H}$ is at least $k$,
  let $u$ and $v$ be two arbitrary nodes of $\mathsf{H}$. We show that
  there are at least $k$ internally vertex disjoint paths between $u$
  and $v$. If $u$ and $v$ are both in $\mathcal{A}$, then there is one
  direct edge between $v$ and $u$ and there are $2k-2$ paths of length
  $2$ between them. If exactly one of $v$ and $u$ is in $\mathcal{A}$,
  e.g., suppose $u \in \mathcal{A}$ and $v \in \mathcal{B}$, then $u$
  is directly connected to $k$ neighbors of $v$. Otherwise, if both
  $u$ and $v$ are in $\mathcal{B}$, then let $p$ be the size of the
  intersection of the neighbors of $v$ and $u$. Note that these
  neighbors are all in $\mathcal{A}$. It is clear that $u$ and $v$
  have exactly $p$ paths of length $2$ between themselves and $k-p$
  paths of lengths $3$, and that these paths are internally vertex
  disjoint.

  To see that $K'_{\CDS}(H) < 2$, first note that each CDS $\tau$ must
  include at least $k+1$ nodes of $\mathcal{A}$. This is because,
  otherwise, there are at least $k$ nodes of $\mathcal{A}$ that are
  not included in $\tau$ and thus, there is a node in
  $\mathcal{B}$---corresponding to a subset of size $k$ of these
  uncovered nodes of $\mathcal{A}$---which is not dominated by
  $\tau$. Thus we have,
  $$\sum_{v \in \mathcal{A}} \sum_{\tau \in \CDS(\mathsf{H}) \atop
    s.t.\; v \in \tau} x_{\tau} 
  \geq (k+1) \cdot \sum_{\tau \in \CDS(\mathsf{H})} x_{\tau}.$$
  On the other hand we have, 
  $$\sum_{v \in \mathcal{A}} \sum_{\tau \in \CDS(\mathsf{H}) \atop
    s.t.\; v \in \tau} x_{\tau} \leq \sum_{v \in \mathcal{A}} 1 =
  |\mathcal{A}| = 2k, $$ and thus we can conclude that $\sum_{\tau \in
    \CDS(\mathsf{H})} x_{\tau} \leq \frac{2k}{k+1} <2$. Since this
  holds for any CDS-Packing of $\mathsf{H}$, we get $K'_{\CDS}(H) <
  2$.
\end{proof}

Note that in the above construction, we have $K_{\CDS}(\mathsf{H}) =
1$ as $\mathsf{H}$ is connected and $K_{\CDS}(\mathsf{H})$ has to be
an integer.

\begin{lemma} \label{lem:badgraph_larger_k}For each large enough $k$ and $\eta \in [4k, 2^{k}]$,
  there exists a sub-graph $H' \subseteq \mathsf{H}$ that has $\eta$ nodes and vertex
  connectivity $k$ but $K'_{\CDS}(H') = O(\frac{k}{\log
    {\eta}})$.
\end{lemma}
\begin{proof}
  Pick an arbitrary $k \geq 64$, fix an $\eta \in [4k, 2^{k}]$ and let
  $\beta = \frac{\log{\eta}}{8}$. Consider a random subset $V_z
  \subseteq V$, where $V_z$ includes all nodes of $\mathcal{A}$ and
  for each node $u \in \mathcal{B}$, $u$ is independently included in
  $V_z$ with probability $p$, where
  $$p = \frac{65 \beta^2}{\binom{2k-\beta}{k}}.$$
  We now look at the sub-graph $H_z$ of $\mathsf{H}$ induced on
  $V_z$. With the same argument as for $\mathsf{H}$, we get that for
  any such $V_z$, the graph $H_z$ has vertex connectivity exactly $k$.
  We show that (a) with probability at least $\frac{1}{2}$, $V_z$ is
  such that $K'_{\CDS}(H_z) < \frac{2k}{\beta} = O(\frac{k}{\log
    {\eta}})$, and (b) with probability at least $\frac{3}{4}$, we
  have $|V_z| \leq \eta$. A union bound then completes the proof.

  \paragraph{Property (a)} We first show that with probability at
  least $\frac{1}{2}$, $V_z$ is such that there does not exist a
  subset of size $\beta$ of the nodes of $\mathcal{A}$ that
  dominates $V_z$.
  % 
  % If $K_{\CDS}(H_z) \geq \frac{2k}{\beta}$, then there exists a
  % subset of size $\beta$ of the nodes of the first layer that
  % dominates $V_z$.
  % 
  For each subset $W \subset \mathcal{A}$ such that $|W|=\beta$,
  there are $\binom{2k-\beta}{k}$ nodes in $\mathcal{B}$ which are
  not dominated by $W$. Thus, for $W$ to dominate $V_z$, none of these
  second layer nodes should be included in $V_z$. The probability for
  this to happen is $$(1-p)^{\binom{2k-\beta}{k}} \leq
  e^{-65\beta^2}$$ There are $\binom{2k}{\beta}$ possibilities for set
  $W$. Hence, using a union bound, the probability that there exists
  such a set $W$ that dominates $V_z$ is at most
  \begin{eqnarray}
    e^{-65\beta^2} \binom{2k}{\beta} \leq e^{-65\beta^2} \cdot (\frac{2ek}{\beta})^{\beta} &\stackrel{(\dagger)}{<}& e^{-65\beta^2} \cdot (\eta^2)^{\beta}\nonumber \\
    &=& e^{-65\beta^2 + 64\beta^2} \leq \frac{1}{2}\nonumber,
  \end{eqnarray}
  where Inequality $(\dagger)$ follows since $64 \leq k \leq
  \frac{\eta}{4}$, which gives $2ek < k^2 <\eta^2$.

  Thus, with probability at least $\frac{1}{2}$, $V_z$ is such that
  each CDS of $H_z$ includes at least $\beta+1$ nodes of
  $\mathcal{A}$. From this, similar to the last part of the proof of
  \Cref{lem:badgraph}, we have that, $\sum_{\tau \in \CDS(\mathsf{H})}
  x_{\tau} \leq \frac{2k}{\beta+1} < \frac{2k}{\beta}$. Since this
  holds for any packing of $H_z$, we get that with probability at
  least $\frac{1}{2}$, $V_z$ is such that $K'_{\CDS}(H_z) <
  \frac{2k}{\beta}$.

  \paragraph{Property (b)} Note that $\mathbb{E}[|V_z|] = 2k +p
  \cdot\binom{2k}{k}$. Substituting $p = \frac{65
    \beta^2}{\binom{2k-\beta}{k}}$ and noting that $\beta \leq
  \frac{k}{2}$, we get
  \begin{eqnarray}
    \mathbb{E}[|V_z|] -2k&=&  p \cdot\binom{2k}{k} = 65\beta^2 \cdot \frac{\binom{2k}{k}}{\binom{2k-\beta}{k}} \nonumber\\
    &=& 65\beta^2 \cdot\frac{2k}{2k-\beta} \cdot \frac{2k-1}{2k-\beta-1} \dots \frac{k+1}{k-\beta+1} \nonumber\\
    &\leq& 65\beta^2 \cdot (1+ \frac{2\beta}{k})^{k}\nonumber\\
    &\leq&  \frac{65\log^2\eta }{64}\cdot \eta^{\frac{1}{4}}    
    \ \leq\ \frac{\eta}{4}.\nonumber
  \end{eqnarray}
  As the second-layer nodes are picked independently, for $\eta$
  sufficiently large, we can apply a Chernoff bound to get $\Pr[|V_z|
  -2k > \frac{\eta}{2}] \leq \frac{1}{4}$. Since $2k \leq
  \frac{\eta}{2}$, we then obtain $\Pr[|V_z| > \eta] \leq
  \frac{1}{4}$. If desired, we can adjust the number of nodes to
  exactly $\eta$ by adding enough nodes in the second layer which are
  each connected to all nodes of the first layer.
 \end{proof}

%Grab message of an arbitrary second layer node v. Key point is that,
%for a successful all-to-all message exchange, this message has to be
%transmitted by at least $k$ nodes of the first layer (during the
%execution). This is because, suppose that message of v is transmitted
%by at most $k-1$ nodes of first layer. Then, there are $k+1$ nodes of
%first layer that did not transmit message of v at all. For these
%$k+1$ nodes of first layer, there are $\binom{k+1}{k} =k+1$ nodes of
%second layer which are only connected to these first layer nodes
%(i.e., don't have any neighbor outside these $k+1$ nodes) Since
%$k>1$, at least one of these second layer nodes is different than
%$v$. Therefore, there is at least one node of second layer that does
%not receive message of $v$ at all. This contradicts with the
%assumption that  all-to-all message exchange is done
%successfully. Hence, so far we know that message of each second node
%layer has to be transmitted by at least $k$ nodes of first
%layer. Since in each round, each first layer node can transmit at
%most one message, and there are $2k$ such nodes, in $T$ rounds we can
%have at most $2kT$ message transmissions by first layer nodes. Now
%$T$ should be large enough so that each messages of second layer
%nodes get at least $k$ such transmissions. There are roughly $2^k$
%second layer nodes. Hence, $2kT > \Theta( 2^{k}) \cdot k$ which
%results that $T= \Omega(2^k) = \Omega(n)$.

%%% Local Variables: 
%%% mode: latex
%%% TeX-master: "Main"
%%% End: 

\section{Missing Proof of \Cref{subsec:packVSthru}: CDS Packing vs. Throughput}
\label{sec:App5}
In this section, we prove \Cref{prop:thruConnection}. For simplicity, we restate the theorem.

\begin{theoremR}{prop:thruConnection} A CDS packing with size $t$ provides a store-and-forward backbone with broadcast throughput $\Omega(t)$ messages per round. Inversely, a store-and-forward broadcast algorithm with throughput $t$ messages per round induces a CDS packing of size at least $t$.
\end{theoremR}
\begin{proof} [Proof of \Cref{prop:thruConnection}]
  First consider a CDS $\tau$ and suppose that the graph induced by
  $\tau$ has diameter $D_{\tau}$. Using $\tau$, we can perform $p$
  broadcasts (or multicast or unicasts) in time $O(p+D_{\tau})$. This
  can be seen as follows: Since $\tau$ is a dominating set, we can
  deliver each message to a node of $\tau$ in at most $p$
  rounds. Because $\tau$ is connected, $O(p + D_{\tau})$ rounds are
  enough to broadcast the messages to all nodes in $\tau$. Finally,
  because $\tau$ is a dominating set, at most $p$ more rounds are
  enough to deliver the messages to all the desired destination
  nodes. Hence, a CDS structure allows for performing broadcasts with
  an (amortized) rate of $\Omega(1)$ messages per round. In other
  words, a CDS can be viewed as a communication backbone with
  throughout $\Omega(1)$ messages per round.

  Consequently, $t$ vertex-disjoint CDS sets form a communication
  backbone with throughput of $\Omega(t)$ messages per
  round. Intuitively, we can use those $t$ vertex-disjoint sets in
  parallel with each other and get throughput of $\Omega(1)$ message
  per round from each of them. For a more formal description, consider
  $p$ broadcasts such that no more than $q$ broadcasts have the same
  source node. We first deliver each messages to a randomly and
  uniformly chosen CDS set. This can be done in time at most $q$. With
  high probability, the number of messages in each CDS is
  $O(\frac{p}{t} + \log n)$ and thus, we can simultaneously broadcast
  messages in time $O(\frac{p}{t}+\log n + D_{\max})$ where $D_{\max}$
  is the maximum diameter of the CDSs. Thus, the total time for
  completing all the broadcasts is $O(q + \frac{p}{t} + \log n+
  D_{\max})$. That is, we can perform the broadcasts with a rate
  (throughput) of $\Omega(\min\{t, p/q\})$. Note that since each
  source can only send one packet per round, if $ q \leq t$, then the
  maximum achievable throughout with any algorithm including network
  coding approaches is at most $q$ packets per round. In other words,
  in that case, the bottleneck is not the communication protocol but
  rather the sources of the messages. As long as no node is the source
  of more than $\Theta(p/t)$ messages, $t$ vertex-disjoint CDS sets
  form a communication backbone with throughput $\Omega(t)$. 
	%For a distributed implementation of the above protocol, nodes only need to know their CDS. Moreover, the algorithm can be adapted broadcast messages that arrive continuously over time.

  Similarly one can see that a CDS packing with size $t$ provides a
  backbone with a throughput of $\Omega(t)$ messages per round. The
  only change with respect to above description is that now each node
  $v$ spends a $x_\tau$-fraction of its time for sending the messages
  assigned to CDS $\tau$ for every $\tau$ such that $v \in
  \tau$. Further, messages are assigned to each CDS $\tau$ with
  probability proportional to $x_\tau$. We remark here that even
  though this scheme provides a backbone with throughput $\Omega(t)$,
  if the weights $x_\tau$ are too small, the outlined time sharing
  might impose a considerably large additive term on the overall time
  for completing the broadcasts. In fact since the number of potential
  CDS sets can be exponential, the time sharing might lead to
  exponentially large additive terms. Note that in the CDS packing we present in \Cref{thm:packing}, each CDS has weight at least
  $\Omega(1/\log n)$ and thus using the partition also leads to an asymptotically optimal throughput for a relatively small number of broadcast messages.  

  Let us now argue that a broadcast protocol with throughput $t$ also
  leads to a CDS packing of size $t$. Suppose that there exists a
  (possibly large enough) number $p$ and a store-and-forward algorithm
  which broadcasts $p$ messages (originating from potentially
  different sources) in $T\leq \frac{p}{t}$ rounds. For each message
  $\sigma$ that is being broadcast, define set $S(\sigma)$ to be the
  set of nodes that send $\sigma$ in some round of the
  algorithm. Clearly $S(\sigma)$ induces a connected sub-graph and
  because every node needs to receive the message $S(\sigma)$ also is
  a dominating set.  For each node $v$ and message $\sigma$, let
  $N_{\sigma}(v)$ be the number of rounds in which node $v$ sends
  message $\sigma$ and let $y_{\sigma}(v)=
  \frac{N_{\sigma}(v)}{T}$. Moreover, for each CDS $\tau$ such that
  $v\in \tau$, let 
  \[
  z_{\tau}(v) = 
  \sum_{\sigma \atop S(\sigma) =\tau \,\land\, v \in \tau} 
  y_{\tau}(v).
  \]
  Finally, let $x_{\tau} = \min_{v\in \tau} \{z_{\tau}(v)\}$. Given
  these parameters, first notice that for each node $v$, we
  have 
  \[
  \sum_{\tau \atop v\in \tau} x_{\tau} \leq \sum_{\tau \atop
    v\in \tau} z_{\tau}(v) = \sum_{\tau \atop v\in \tau} \sum_{\sigma
    \atop S(\sigma)=\tau} y_{\sigma}(v) = \sum_{\sigma}
  \frac{N_{\sigma}(v)}{T} \stackrel{(\dagger)}{\leq} 1.
  \]
  Here, Inequality ($\dagger$) is because in each round, node $v$ can
  send at most one message and thus, $\sum_{\sigma} N_{\sigma}(v) \leq
  T$. On the other hand, we show that $\sum_{\tau} x_{\tau} \geq
  \frac{p}{T} = \Omega(t)$. For this purpose, consider a CDS $\tau$
  and let $u^*$ be a node such that $z_{\tau}(u^*) = x_{\tau}$. Since
  each message $\sigma$ such that $S(\sigma) = \tau$ is sent at least
  once by $u^*$, we have
\begin{eqnarray}
\sum_{\sigma \atop S(\sigma) = \tau} 1 \leq  \sum_{\sigma \atop S(\sigma) = \tau} N_{\sigma}(u^*) =& \sum_{\sigma \atop S(\sigma) = \tau} y_{\sigma}(u^*) \cdot T\nonumber \\
=& z_{\tau}(u^*) \cdot T= x_{\tau} \cdot T \nonumber
\end{eqnarray} 
Moreover, we have that $$p =\sum_{\sigma} 1 = \sum_{\tau} \sum_{\sigma \atop S(\sigma) = \tau} 1  \leq \sum_{\tau} x_{\tau} \cdot T$$ 
Thus, $\sum_{\tau} x_{\tau} \geq \frac{p}{T}$. Since $T\leq \frac{p}{t}$, we get that $\sum_{\tau} x_{\tau} \geq t$.
\end{proof}

%%% Local Variables: 
%%% mode: latex
%%% TeX-master: "Main"
%%% End: 

\section{Missing Proofs of \Cref{sec:threshold}}\label{app7}

\begin{proof}[Proof of \Cref{lem:dom}] 
  Note that since $G$ has vertex connectivity $k$, each node $v \in V$ has at least $k$ neighbors.
  Therefore, each virtual node $\nu \in \mathcal{V}$ has in expectation at least $\frac{kp\NumOfLayers}{2\NumOfLayers} = \alpha\sqrt{k}\log n/2 = \Omega(\log n)$
  virtual neighbors in $\mathcal{S}_{\frac{\NumOfLayers}{2}}$. For every $\lambda>0$ and a sufficiently large $\alpha>0$, a standard Chernoff argument combined with union bound
  over all virtual nodes implies that each node has at least one neighbor in $\mathcal{S}_{\frac{\NumOfLayers}{2}}$, w.h.p. 
\end{proof}

\begin{proof}[Proof of \Cref{lem:FML}]
  Assume $\mathcal{S}_{\NumOfLayers/2}\in \mathcal{D}$. Then, for the first part of lemma, note that since $\mathcal{S}_{\NumOfLayers/2}\in \mathcal{D}$ and for each layer $\ell \geq \NumOfLayers/2$, $\mathcal{S}_\frac{\NumOfLayers}{2} \subseteq S_\ell$, we get that $S_{\ell}$ dominates set $\mathcal{V}$. Thus, each virtual sampled node in layer $\ell+1$ has a neighbor in $S_\ell$ which means that each connected component of $\mathcal{G}[\mathcal{S}_{\ell+1}]$ contains at least one connected component of $\mathcal{G}[\mathcal{S}_{\ell}]$. Hence, $N_{\ell+1} \leq N_{\ell}$, which also means that $M_{\ell+1} \leq M_\ell$.
	
	For the second part of the lemma, fix an arbitrary $\mathcal{T}\in \mathcal{D}_\ell$ and suppose that $\mathcal{S}_\ell=\mathcal{T}$. Let $X$ be the number of components of $\mathcal{G}[\mathcal{S}_{\ell}]$ for which there is
  at least one connector path in layer $\ell+1$. Note that each such component gets merged with at least one other component. By Lemma \ref{lem:paths}, the
  expectation of $X$ is at least $\E[X]\geq N_\ell/2$. Let $q$ be
  the probability that $X\geq N_\ell/4$. We have
  $\E[X]=\sum_{x=1}^{N_\ell}x\Pr(X=x)<q\cdot N_\ell +
  (1-q)\frac{N_\ell}{4}$. Together with the upper bound $N_\ell/2$
  on $\E[X]$, this gives $q> 1/3$. If $X\geq N_\ell/4$, at least
  $1/4$ of the components of $\mathcal{G}[\mathcal{S}_{\ell}]$ have at least one connector path in layer $\ell+1$, i.e., are connected to some other
  component of $\mathcal{G}[\mathcal{S}_{\ell}]$. In that case, $N_{\ell+1}\leq \frac{7}{8}\cdot N_{\ell}$. Thus, $\Pr\big[M_{{\ell}+1}
  \leq \frac{9}{10} \cdot M_{\ell} \,|\,
  \mathcal{S}_\ell=\mathcal{T}\big]\geq 1/3$. 
\end{proof}

\begin{proof} [Proof of \Cref{thm:threshold}]
  From \Cref{lem:dom}, we know that with high probability, $\mathcal{S}_\frac{\NumOfLayers}{2}$ is a dominating set. Assuming that $\mathcal{S}_\frac{\NumOfLayers}{2}$ is a dominating set, we can use \Cref{lem:FML}. Thus, we get that with addition of sampled virtual nodes of each layer $\ell \geq \NumOfLayers/2$, the number of components of the sampled virtual nodes goes down by a constant factor, with probability at least $1/3$ and
  independent of previous layers. Hence, by a standard Chernoff argument, we obtain that after $\NumOfLayers=\lambda \log n$ layers, with high probability, $\mathcal{S}$ is a connected dominating set of graph $\mathcal{G}$. Thus, $S$ is a connected dominating set of graph $G$, and hence, in particular, $G[S]$ is connected.
\end{proof}

We conclude this section by presenting Observation \ref{obs:neg-side}, which shows that the bound of \Cref{thm:threshold} is optimal up to an $O(\sqrt{\log n})$ factor. 

\begin{observationR}{obs:neg-side} For every $k$ and $n$, there
  exists an $n$-node graph $G$ with vertex connectivity $k$ such that
  if we independently sample vertices with probability $p\leq
  \sqrt{\frac{\log(n/2k)}{2k}}$, then the subgraph induced by the
  sampled nodes is disconnected with probability at least $1/2$.
\end{observationR}

\begin{proof}[Proof of Observation \ref{obs:neg-side}]
  Consider a graph $G$ composed of a chain of $\frac{n}{k}$ cliques, each
  of size $k$, where each two consecutive cliques on the chain are
  connected by a matching of size $k$. For simplicity, assume that
  $\frac{n}{k}$ is an even integer. Clearly, this graph has vertex
  connectivity $k$. Suppose that the sampling probability $p$ is less than $\sqrt{\log (\frac{n}{2k})/(2k)}$. We show that, with
  probability at least $1/2$, the graph induced on the sampled vertices is disconnected.
  
  Let us number the cliques from $1$ to $\frac{n}{k}$. First note that
  some nodes of each clique are sampled w.h.p. Let $J =
  \set{j | j \in [1,\frac{n}{k}] \land j \equiv 0 \pmod 2 }$, i.e., the
  set of even numbers in the range $[1,\frac{n}{k}]$. If the sampled graph is connected, then for each $j \in J$, at least on of the edges of the
  matching connecting clique $j-1$ with clique $j$ has to be in the
  sampled graph. For this, there have to be two adjacent
  nodes $u$ from clique $j-1$ and $v$ from clique $j$ that are both sampled. We call such a pair of nodes
  \emph{connecting}. For each $j \in J$, there are $k$ possible
  choices for a connecting pair and each choice has probability
  $\frac{1}{s^2}$ to be sampled (to be connecting). Moreover,
  these pairs are vertex disjoint. Thus, for each $j \in J$, the
  probability that the sampled graph has a connecting pair between cliques $j-1$
  and $j$ is 
  \[
  1-\left (1-\frac{1}{{s}^2}\right)^{k}
  \ \leq\ 1 -4^{- k/s^2}
  \ <\ 1 -4^{- \log(\frac{n}{2k})/2}
  \ =\ 1 -\frac{2k}{n}.
  \]
  Since $J$ only contains even values, the pairs of consequent cliques
  that we consider are disjoint. Thus, the events whether they have a
  connecting pair or not are independent. Therefore, the probability
  that we have at least one connecting pair for each value of $j \in
  J$ is at most
  $\big(1 - \frac{2k}{n}\big)^{|J|} = \big(1 - \frac{2k}{n}\big)^{\frac{n}{2k}}
  <\frac{1}{e} < 1/2 $.
\end{proof}

\section{Missing Proofs of \Cref{sec:general}}\label{app8}
\begin{proof}[Proof of \Cref{lem:layer1dom} ]
For each node real node $v$, $v$ has at least $k$ real neighbors and in these $k$ real neighbors, in expectation at least $\frac{kq}{2t} = \Omega(\log n)$ virtual nodes are sampled, have layer number $1$ to $\frac{\NumOfLayers}{2}$, and join class $i$, i.e., are in $\mathcal{S}^i_{\frac{\NumOfLayers}{2}}$. Thus, the claim follows from a standard Chernoff argument combined with a union bound over all choices of $v$ and over all classes.
\end{proof}

\begin{proof}[Proof of \Cref{lem:CAL}]
As in the proof of \Cref{lem:paths}, using Menger's Theorem on graph $G$ along with the Domination Lemma (\Cref{lem:layer1dom}), we can infer that if connected component $\mathcal{C}$ is not alone in its class, then $\Psi(\mathcal{C})$ has at least $k$ internally vertex-disjoint potential connector paths connecting it to $\Psi(\mathcal{S}^i_\ell \setminus \mathcal{C})$, in graph $G$. Using rules (D) and (E) in \Cref{subsec:Connectors}, we get $k$ internally vertex-disjoint potential connector paths on the virtual nodes of layer $l+1$. It is clear that during the transition from the real nodes to the virtual nodes, the potential connector paths remain internally vertex-disjoint. Now, for each fixed potential connector path on the virtual nodes, the probability that the internal nodes of this path are sampled is at least $q^2$. Hence, in expectation, $\mathcal{C}$ has $kq^2$ internally vertex-disjoint connector paths (on virtual nodes). A simple application of Markov's inequality shows that with probability at least $1/2$, $\mathcal{C}$ has at least $t=\Omega(kq^2)$ internally vertex-disjoint connector paths.
\end{proof}

\begin{proof}[Proof of \Cref{lem:FML2}] Consider a class $i$ such that $N^i_\ell \geq 2$ and let $\mathcal{C}_1$ be a connected component of $\mathcal{G}[\mathcal{S}^i_\ell]$. We say component $\mathcal{C}_1$ is \emph{good} if for at least one connector path $p$ of $\mathcal{C}_1$, all internal nodes of $p$ --- one or two nodes depending on whether $p$ is short or long --- join class $i$. Note that if $\mathcal{C}_1$ is good, then it gets connected to another component of $\mathcal{G}[\mathcal{S}^i_\ell]$. In order to prove the lemma, we show that with probability at least $1/2$, at least $\frac{M_\ell}{3}$ connected components (summed up over all classes) are good. This is achieved using a simple accounting method by considering the number of remaining connector paths as the budget. 

We know that with probability at least $1/2$, initially we have a budget of $M_\ell \cdot t$. We will show that the greedy algorithm spends this budget in a manner that at the end, we get $M_{\ell+1}\leq \frac{5}{6} M_\ell$.

In each step of each of stages I or II, if respectively a type-$1$ node or a type-$3$ nodes and some associated type-$2$ nodes join a class, then at least $\Delta$ components become good where $\Delta$ is defined as explained in the algorithm description. We show that in that case, we remove at most $3\Delta t$ connector paths in the related bookkeeping part. Thus, in the accounting argument, we get that at most $3\Delta \cdot t$ amount of budget is spent and $\Delta$ components become good. Hence, in total over all steps, at least $\frac{M_\ell}{3}$ components become good. 

Let us first check the case of short connector paths, which is performed in stage I. Let $v$ be the new type-$1$ node under consideration in this step and suppose that the related $\Delta \geq 2$, and node $v$ joins class $i^*$. For class $i^*$, we remove all paths of all connected components of $i^*$ that have $v$ on their short connector paths. This includes $\Delta$ such connected components, and $t$ connector paths for each such component. Thus, in total we remove at most $\Delta \cdot t$ connector paths of components of class $i^*$. For each class $i \neq i^*$, we remove at most $\Delta$ connector paths. This is because $v$ can be on short connector paths of at most $\Delta$ components, at most once for each such component. These are respectively because of definition of $\Delta$ and due to internally vertex-disjointedness of connector paths of each component. There are less than $t$ classes other than $i^*$, so in total over all classes other than $i^*$, we remove at most $\Delta \cdot t$ connector paths. Therefore, we can conclude that the total amount of decrease in budget is at most $2\Delta \cdot t$.

Now we check the case of long connector paths, performed in stage II. Suppose that in this step, we are working on a type-$3$ new node $u$, it has $\Delta \geq 1$, and we assign node $u$ and associated type-$2$ new nodes $v_1$, $\dots$, $v_{\Delta}$ to class $i^*$. It follows from \Cref{prop:longpaths} that nodes $v_1$, $\dots$, $v_{\Delta}$ are not on long connector paths of components of class $i^*$ other than the $\Delta$ components which have long paths through $u$. Thus, any connector path of class $i^*$ that goes through any of $v_1$ to $v_{\Delta}$ also goes through $u$. For each component of class $i^*$ that has a long connector path through $u$, we remove all the connector paths. By definition of $\Delta$, there are $\Delta$ such components and from each such component, we remove at most $t$ paths. Hence, the number of such connector paths that are removed is at most $\Delta \cdot t$. On the other hand, for each class $i\neq i^*$, we remove at most $2\Delta$ connector paths. This is because by definition of $\Delta$, removing just node $u$ removes at most $\Delta$ long paths from each class. Moreover, because of \Cref{prop:longpaths} and internally vertex-disjointedness of connector paths of each component, removing each type-$2$ node $v_{j}$ (where $j \in \{1,2, \dots, \Delta\}$) removes at most one long connector path of one connected component of class $i \neq i^*$. Over all classes $i \neq i^*$, in total we remove at most $2\Delta \cdot t$ connector paths. Hence, when summed up with removed connector paths related to class $i^*$, we get that the total amount of decrease in the budget is at most $3\Delta \cdot t$.
\end{proof}

\section{Missing Proofs of \Cref{sec:partition}}\label{app9}

\begin{proof}[Proof of \Cref{lem:CAL-Part-Large}]
  Fix a class $i$. We start by studying the connected components of
  $G[V^i_1]$. Each node has probability $\frac{1}{tL} = \frac{\log
    n}{\delta\lambda k}$ to be in $V^i_1$. Therefore,
  $\E[|V^i_1|] = \frac{n\log n}{\delta\lambda k}$. Using a
  Chernoff bound we get that for $\delta$ sufficiently small, w.h.p.,
  $|V^i_1| \leq \frac{2n\log n}{\delta\lambda k}$. Moreover, since
  each node $v \in V$ has at least $k$ neighbors in $G$, the expected
  number of neighbors of $v$ in $V^i_1$ is at least $\Omega(\log
  n)$. Using another Chernoff bound (and $\delta$ sufficiently small)
  and then a union bound over all choices of $v$, we get that w.h.p.,
  each node $v \in V$ has $\Omega(\log n)$ neighbors in
  $V^i_1$. Therefore, in particular, each node $v \in V^i_1$ has
  $\Omega(\log n)$ neighbors in $V^i_1$. In other words, w.h.p., the
  degree of each node in $G[V^i_1]$ is $\Omega(\log n)$. Thus, w.h.p.,
  $G[V^i_1]$ has at most $\frac{2n\log n}{\delta\lambda k}$ nodes, each of
  degree $\Omega(\log n)$. Therefore, for an appropriate
  choice of the constant $\delta$, w.h.p., the number of connected
  components of $G[V^i_1]$ is at most $\eps\cdot\frac{n}{k}$ for a given
  constant $\eps>0$. Let $\Sigma^i$ be the set of all connected components of
  $G[V^i_1]$.
	
  We call each nonempty set $A$ which is a strict subset of $\Sigma^i$, i.e.,
  $A\subset \Sigma^i$, a \emph{component cut} of $G[V^i_1]$. Since, w.h.p., we
  have $|\Sigma^i| \leq \eps\frac{n}{k}$, the number of component cuts of
  $G[V^i_1]$ is at most $2^{\eps n/k}$, w.h.p.
	
	For each layer $\ell \in [2, \frac{L}{2}]$ and each component cut
  $A$ of $G[V^i_1]$, we say $A$ is \emph{$\ell$-rich} if there are at least $\frac{k}{8}$ internally
  vertex-disjoint paths $p$ which satisfy the following conditions:
  (1) $p$ has one end point $s \in A$ and the other endpoint $t \in
  S\setminus A$, (2) $p$ has at most two internal nodes, (3) if $p$
  has two internal nodes and has the form $s$, $u$, $w$, $t$, then $u$
  does not have a neighbor in $S-A$ and $w$ does not have a neighbor
  in $S$, (4) all internal nodes of $p$ are in layers $[\ell+1, L]$. 
	We first show that with high probability, for each layer $\ell \in [2, \frac{L}{2}]$ and each component cut
  $A$ of $G[V^i_1]$, $A$ is \emph{$\ell$-rich}.

  Consider an arbitrary component cut $A$ of $G[V^i_1]$. Since graph
  $G$ is $k$-vertex connected and because $V^i_1$ is a dominating set
  (cf.\ \Cref{lem:layer1dom}), there are at least $k$ internally
  vertex-disjoint paths which satisfy conditions (1) and (2). The
  details of this argument are similar to the first part of the
  proof of \Cref{lem:paths}. % and is thus skipped.
  Each of these $k$ paths can be \emph{trimmed} to also satisfy condition
  (3). To see this, consider a path $p$ as
  described in condition (3). If $v$ has a neighbor in $S\setminus A$ then
  there is a trimmed path $p'$ from some node in $A$ to $v$ to some
  node in $S\setminus A$, which satisfies (3). Similarly, if $w$ has a neighbor in $A$, then there is a
  trimmed path $p'$ from $A$ to $w$ to some node in $S\setminus A$, which
  satisfies (3). In either case, the trimmed path $p'$ remains internally vertex-disjoint from
  the other $k-1$ paths.
  % We refer to this step as the \emph{trimming} step.

  Thus far, we have found $k$ internally vertex-disjoint paths of
  component cut $A$ satisfying conditions (1) to (3). Now, each
  internal node $u$ on each of these $k$ paths has probability at
  least $\frac{1}{2}$ to be in one of layers $[\ell+1, L]$. Formally,
  this is because, the layer number of $u$ is chosen randomly and so
  far, the only information exposed about $u$ is that it is not in
  $V^i_1$. Let $\layer(u)$ be the layer number of $u$. We get
  	\begin{eqnarray*}
    \Pr[\layer(u) \in [\ell+1, L]| u \notin V^i_1] & = & \frac{\Pr[\layer(u) \in [\ell+1, L] \land u \notin V^i_1]}{\Pr[u \notin V^i_1]}\\
    & = & \frac{\Pr[\layer(u) \in [\ell+1, L]]}{\Pr[u \notin V^i_1]} = \frac{\frac{L-\ell}{L}}{1-\frac{1}{Lt}} \geq \frac{L-\ell}{L} \geq \frac{1}{2},
  \end{eqnarray*}
  where the last inequality holds because $\ell \leq \frac{L}{2}$. 
	Hence, each internal node $u$ on each path $p$ out of
  the $k$ internally vertex-disjoint paths for $A$ has probability at
  least $1/2$ to be in layer $[\ell+1, L]$. The expected number of
  paths that also satisfy condition (4) therefore is at least
  $\frac{k}{4}$. Since the paths are internally vertex-disjoint and
  layer numbers are chosen independently, we can use a Chernoff bound
  and conclude that with probability at least $1 - e^{-k/32}$, at
  least $\frac{k}{8}$ of the paths satisfy condition
  (4). Consequently, with probability at least $1 - e^{-k/32}$, the
  component cut $A$ of $G[V^i_1]$ is $\ell$-rich.

  Now there are at most $2^{\eps n/k}$ component cuts for
  $G[V^i_1]$. Thus, using a union bound, the probability that there
  exists one of them that is not $\ell$-rich is at most $2^{\eps n/k}
  \cdot e^{-k/48}$. Since $k = \Omega(\sqrt{n})$, for $\eps$ small
  enough, this probability is less than
  $2^{-\Omega(\sqrt{n})}=2^{-\omega(\log n)}$. Hence, w.h.p., each
  component cut of $G[V^i_1]$ is $\ell$-rich. Using a union bound over
  all choices of $\ell \in [2, \frac{L}{2}]$, we can also conclude that
  for each such $\ell$, each component cut of $G[V^i_1]$ is $\ell$-rich.

  We are now ready to show that each connected component $\mathcal{C}$
  of $G[V^i_{\ell}]$ has at least $\Omega(k/\NumOfLayers^2)$ internally
  vertex-disjoint connector paths. Note that each component $\mathcal{C}$ is composed of
  a subset of the connected components of $G[V^i_1]$ and some nodes of layers $2,\dots,\ell$. Hence, there is a component cut $A'$ of
  $G[V^i_1]$ such that $\mathcal{C} \cap V^i_1 = A'$. Since the component
  cut $A'$ is $\ell$-rich w.h.p., there are at least $\frac{k}{8}$
  paths which satisfy conditions (1) to (4) above. Each such
  path $p$ satisfies conditions (A), (B). However, path $p$ might not directly satisfy condition (C). This is
  because, it is possible that $p$ is defined as $s$, $v$, $w$, $t$
  which satisfies condition (3) but in graph $G[V^i_\ell]$, node $v$ has
  a neighbor in $W^i_\ell - \mathcal{C}$ or node $w$ has a neighbor in
  $\mathcal{C}$. Though, in either case, we can get a path
  $p'$ of length $2$ which satisfies conditions (A), (B), and
  (C) with just one internal node, either $v$ or $w$. Note that
  $p'$ still remains internally-vertex-disjoint from the other paths. Moreover, the internal node of $p'$ is in layers
  $[\ell+1, L]$.

  So far, we have found $\frac{k}{8}$ potential connectors for
  $\mathcal{C}$ which have their internal nodes in layers in
  $[\ell+1, L]$. For each internal node $v$ of any of these paths, it
  holds that $v$ is in layer $\ell+1$ with probability greater than
  $\frac{1}{L}$. Formally this is true because, given that node $v$ is
  in a layer in $[\ell+1, L]$, the exact layer number of $v$ can be
  chosen after making all the decisions for the first $\ell$
  layers. Therefore, each of the paths that we have found so far has
  probability at least $\frac{1}{L^2}$ to satisfy the corresponding rule
  conditions (D) or (E) such that it would be a connector
  path for $\mathcal{C}$.
  % it holds that $v \notin W^i_{\ell}$. Recall that layer numbers are
  % chosen randomly and independently of all other events. Thus,
  % regardless of how classes are chosen, the following statement
  % holds: conditioned on the event that $v \notin W^i_{\ell}$, the
  % probability that node $v$ is in layer ${\ell}+1$ is at least
  % $\frac{1}{\NumOfLayers}$. Finally, conditioned on the event that
  % $v$ is in layer ${\ell}+1$, it has probability $\frac{1}{3}$ for
  % being in each type.
  % 
  Hence, the expected number of connector paths of $\mathcal{C}$ is at
  least $\frac{k}{8\NumOfLayers^2}$. Since the potential connectors we
  found (which have internal nodes in layers $\ell+1$ to $L$) are
  internally vertex-disjoint, the events of them being connector paths
  for $\mathcal{C}$ are independent. Thus, using a Chernoff bound
  we get that w.h.p., $\mathcal{C}$ has at least
  $\Omega(k/\NumOfLayers^2)=\Omega(k/\log^2n)=\omega(\log n)$
  connector paths.  A union bound over all connected components of
  class $i$, over all layers $\ell \in [2, \frac{L}{2}]$, and over all
  choices of the class $i$ completes the proof.
\end{proof}

\end{document}